\documentclass[sigconf,nonacm,metadata=false]{acmart}
\renewcommand\footnotetextcopyrightpermission[1]{} 
\usepackage{tikz}
\usetikzlibrary{arrows.meta,bending}
\usepackage{fancyhdr}
    \usepackage{indentfirst}
	\usepackage{graphicx} 
	\usepackage{balance}
    
	\usepackage{amssymb,amsmath,amsthm}
        \newtheorem{theorem}{Theorem}
        \newtheorem{lemma}{Lemma}
	\usepackage{enumitem}  
	\usepackage{multicol}
	
	\usepackage{multirow, rotating, wasysym, url}
	\usepackage[tight]{subfigure}
	\usepackage[ruled,linesnumbered,vlined]{algorithm2e}
	\usepackage{color}
	\usepackage{soul}
	\usepackage{bbm}
	\usepackage[margin=0.05in]{caption}
	\usepackage{enumerate}
	\usepackage{framed}
	\usepackage{lipsum}
	\usepackage[normalem]{ulem}
	\usepackage{pifont}
        \usepackage[framemethod=tikz]{mdframed}

        \SetKwInput{KwInput}{Input}                % Set the Input
        \SetKwInput{KwOutput}{Output}              % set the Output

        \newcommand{\reviewA}{}
        \newcommand{\reviewB}{}
        \newcommand{\reviewC}{}
        \definecolor{reviewa}{HTML}{F00000}
        
        \definecolor{reviewb}{HTML}{FF00FF}
        
        \definecolor{reviewc}{HTML}{00F000}
        
        \definecolor{reviewd}{HTML}{00F0FF}
        
        \definecolor{reviewe}{HTML}{0000F0}

	\newcommand{\presec}{}
	\newcommand{\postsec}{}
	
	\newcommand{\presub}{}
	\newcommand{\postsub}{}

	\newcommand{\eg}{\textit{e.g.}}

	\mathchardef\Gamma="0100 \mathchardef\Delta="0101
\mathchardef\Theta="0102 \mathchardef\Lambda="0103
\mathchardef\Xi="0104 \mathchardef\Pi="0105
\mathchardef\Sigma="0106 \mathchardef\Upsilon="0107
\mathchardef\Phi="0108 \mathchardef\Psi="0109
\mathchardef\Omega="010A

\newcommand{\outline}[1]{}%{\textbf{#1}}

\usepackage{xspace}
\usepackage{url}
\usepackage{graphicx}
\usepackage{latexsym}
\usepackage{amsfonts}
\usepackage{psfrag}
\usepackage{wrapfig}
\usepackage{comment}
\usepackage{alltt}
\usepackage{color}

\newcommand{\Comment}[1]{}

    \definecolor{reviewA}{HTML}{0000FF}
    \definecolor{reviewB}{HTML}{FF0000}
    \definecolor{reviewC}{HTML}{FF00FF}
	\definecolor{reviewD}{HTML}{00CC00}

	\newcommand{\bbb}{\noindent\textbf}
	
	\definecolor{zyd}{RGB}{50,50,200}

	\definecolor{yt}{RGB}{0,166,0}
	
	\definecolor{zz}{RGB}{0,200,200}
	
	\definecolor{hwc}{RGB}{200,0,200}
	
	\definecolor{zyk}{RGB}{200,50,50}

	\definecolor{remove}{RGB}{200,200,200}
	\newcommand{\remove}[1]{}

	\newcommand{\update}{}

	\newcommand{\algo}{GapFilter\xspace}

        \newtheorem{definition}{Definition}
        
\begin{document}

\title{Detecting Flow Gaps in Data Streams}
% \maketitle
% \author{
% \IEEEauthorblockN{Siyuan Dong, Yuxuan Tian, Wenhan Ma, Tong Yang, Chenye Zhang, Yuhan Wu, Kaicheng Yang, Yaojing Wang}
% \IEEEauthorblockA{School of Computer Science, Peking University}
% }
\author{Siyuan Dong}
\affiliation{%
  \institution{Peking University}
  % \city{Beijing}
  \country{China}
}
\email{dongsiyuan@pku.edu.cn}

\author{Yuxuan Tian}
\affiliation{%
  \institution{Peking University}
  % \city{Beijing}
  \country{China}
}
\email{tianyuxuan@stu.pku.edu.cn}

\author{Wenhan Ma}
\affiliation{
  \institution{Peking University}
  \country{China}
}
\email{mawenhan@stu.pku.edu.cn}

\author{Tong Yang}
%\authornote{Corresponding author (e-mail: yangtongemail@gmail.com).}
\affiliation{%
  \institution{Peking University}
  % \city{Beijing}
  \country{China}
}
\email{yangtong@pku.edu.cn}

\author{Chenye Zhang}
\affiliation{%
  \institution{Peking University}
  % \city{Beijing}
  \country{China}
}
\email{z_cy@pku.edu.cn}
% \email{xuehanyu@pku.edu.cn}

\author{Yuhan Wu}
\affiliation{%
  % \institution{Boston University}
  \institution{Peking University}
  % \city{Beijing}
  \country{China}
}
% \email{pqchen99@bu.edu}
\email{yuhan.wu@pku.edu.cn}

\author{Kaicheng Yang}
\affiliation{%
  \institution{Peking University}
  % \city{Beijing}
  \country{China}
}
\email{ykc@pku.edu.cn}

\author{Yaojing Wang}
\affiliation{%
  \institution{Nanjing University}
  % \city{Beijing}
  \country{China}
}
\email{ghost.wyj@gmail.com}

% \ccsdesc[500]{Theory of computation~Sketching and sampling}
% \ccsdesc[300]{Theory of computation~Distributed algorithms}

% \keywords{interruption detection, real-time, sketch, sequence number}

% \keywords{finding global top-$K$, top-$K$-fairness, sketch, double-anonymity}

    % \balance
    % \input{coverLetter}
    % \maketitle

    % \presec
\begin{abstract}
% \postsec

Data stream monitoring is a crucial task which has a wide range of applications.
% Finding top-$K$ frequent items has been a hot topic in data stream processing in recent years, which has a wide range of applications.
%
The majority of existing research in this area can be broadly classified into two types, monitoring \textit{value} sum and monitoring \textit{value} cardinality. In this paper, we define a third type, monitoring \textit{value} variation, which can help us detect \textit{flow gaps} in data streams. To realize this function, we propose \algo, leveraging the idea of Sketch for achieving speed and accuracy. To the best of our knowledge, this is the first work to detect \textit{flow gaps} in data streams. Two key ideas of our work are the \textit{similarity absorption} technique and the $civilian$-$suspect$ mechanism. The \textit{similarity absorption} technique helps in reducing memory usage and enhancing speed, while the $civilian$-$suspect$ mechanism further boosts accuracy by organically integrating broad monitoring of overall flows with meticulous monitoring of suspicious flows.
% detecting and focusing on the suspicious items. 
% 
We have developed two versions of \algo. Speed-Oriented \algo (\algo-SO) emphasizes speed while maintaining satisfactory accuracy. Accuracy-Oriented \algo (\algo-AO) prioritizes accuracy while ensuring considerable speed. 
We provide a theoretical proof demonstrating that \algo secures high accuracy with minimal memory usage. Further, extensive experiments were conducted to assess the accuracy and speed of our algorithms. The results reveal that \algo-AO requires, on average, 1/32 of the memory to match the accuracy of the Straw-man solution. \algo-SO operates at a speed 3 times faster than the Straw-man solution.
All associated source code has been open-sourced and is available on GitHub\cite{link}.

\end{abstract}

% However, most of existing sketch algorithms focuses on finding local top-$K$ in a single data stream.
%
% In this paper, we work on finding global top-$K$ in multiple disjoint data streams.
% %
% We find that directly deploying prior sketch algorithms is often unfair under global scenarios, which will degrade the accuracy of global top-$K$.
% %
% We define \textit{\fairness{}} and show that it is important for finding global top-$K$.
% %
% To achieve \fairness{}, we propose a new sketch framework, called the Double-Anonymous sketch. 
% %
% The process of finding global top-$K$ items is similar to that of paper reviewing and democratic elections. 
% %
% In these scenarios, double-anonymity is often an effective strategy to achieve \fairness{}.
% %
% We also propose two techniques, hot panning, and early freezing, to further improve the accuracy. 
% %
% We theoretically prove that the Double-Anonymous sketch achieves \fairness while keeping high accuracy.
% %
% We perform extensive experiments to verify \fairness{} in the scenario of disjoint data streams.
% %
% The experimental results show that the Double-Anonymous sketch’s error is up to 129 times ($60$ times on average) smaller than the state-of-the-art.
% %
% All the related source code is open-sourced and available at Github.
%
    \maketitle
	
    \presec
\section{Introduction}
\label{sec:intro}

\subsection{Background and Motivation}
\postsub
\label{intro:background}

\begin{sloppypar}

% \MainIdea{Data stream monitoring is important. what is data stream?}
Data stream monitoring has become instrumental in deriving valuable insights for a myriad of applications, such as anomaly detection \cite{pang2014anomaly,eswaran2018spotlight,miao2022burstsketch,manzoor2016fast}, network measurement \cite{namkung2022sketchlib,ding2023bitsense}, predictive maintenance \cite{zenisek2019machine,sahal2020big,wolfartsberger2020data}, and customer behavior analysis \cite{kwan2005customer,alfian2019customer}, among others \cite{divan2020architecture,schweller2007reversible,kargupta2006board}.
% is a crucial task in industry. 
These tasks present significant challenges due to the high volume and velocity of data streams coupled with the constraints on memory and processing time in practical applications.  
A data stream is a series of items, each of which is a \textit{key-value} pair. Items that share the same \textit{key} compose a flow, and the \textit{key} can be considered as the flow ID\footnote{A flow ID is typically defined as a  five-tuple: source IP address, destination IP address, source port, destination port, and protocol.}. The \textit{value} is the metric that needs to be monitored. All flows' items are intermingled together in a data stream ($e.g., \ DS=\{ \langle a,3 \rangle ,  \langle a,2 \rangle ,  \langle b,5 \rangle ,  \langle d,1 \rangle ,  \langle a,4 \rangle , \ldots\}$). 
% In this paper, we define the number of items in a flow as the flow size, or item frequency. A large flow has a relatively large flow size, while a small flow has a relatively small flow size. An abnormal flow is defined as a flow that suffers from problems like item loss or out-of-order. Problem-free flows are called normal flows.

% \MainIdea{prior works}
Sketch, a probabilistic data structure, has gained widespread use in data stream monitoring due to its impressive speed and performance within the confines of limited memory. As of now, most tasks related to data stream monitoring can be categorized into two types:

\begin{itemize}

% (1) 
    % \item Monitoring <$fid$>, a.k.a. cardinality estimation. It needs to maintain the set of the IDs of flows (distinct elements) in the data stream. Some examples of related scenarios are keeping track of distinct search queries ($e.g.$, on \textit{google.com}), finding out unique users who click into a website, and monitoring network traffic. Classical works are Flajolet-Martin (FM) Sketch [??] and HyperLogLog [??]. 

% (2) 
    % \item \dsy{SELECT} 

    % \dsy{FROM} $Data\_Stream$
    
    % \dsy{WHERE} $flow\_id=key$

    % \vspace{0.03in}
    
    \item \bbb{Monitoring the sum of \textit{value} per \textit{key}.} Such tasks record \textit{value} sum of every flow in the data stream, with potential application in anomaly detection, healthcare analysis, social media analysis, and so on.
    % a variety of important problems such as detecting heavy hitter, detecting heavy changes, estimating frequency distribution, and so on. 
    Classical works include CM Sketch \cite{cormode2005improved}, CU Sketch \cite{estan2003new}, and Count Sketch \cite{charikar2002finding}.

% \vspace{0.1in}
% (3) 
    % \item \dsy{SELECT DISTINCT} \textit{value}

    % \dsy{FROM} $Data\_Stream$
    
    % \dsy{WHERE} $key=flow\_id$

    % \vspace{0.03in}

    \item \bbb{Monitoring the cardinality of \textit{value} per \textit{key}.} Such tasks record \textit{value} cardinality of every flow in the data stream. For instance, the \textit{key} could represent the source address, while the \textit{value} could represent the destination address. They are particularly useful in anti-attack scenarios, such as DDoS or superspreaders.
    % , where such detailed information is required. 
    CSM Sketch\cite{li2012per} proposes the randomized counter sharing scheme to solve the problem.
    % {\color{blue}
    % These tasks involve recording the \textit{value} cardinality of each flow in the data stream. For instance, the \textit{key} could represent the source address, while the \textit{value} could signify the destination address. This approach is particularly useful in anti-attack scenarios, such as DDoS or superspreaders, where such detailed information is required. CSM Sketch\cite{li2012per} proposes the randomized counter sharing scheme to solve the problem.
    % }
    
    % fulfill the task.

    % Classical works are Dedose and Super Spreader. 

\end{itemize}

% Nevertheless, none of the above scenarios includes the interruption in the data stream, which is one of the primary causes of decrease in service quality of applications. Because of the variability, high  speed and large volume of data streams, we need to detect interventions in a short period using limited memory, in order to ease the burden of the Internet as much as possible. 

% \vspace{0.1in}

% \MainIdea{our proposed task}

In this paper, we define a third task type: 
\begin{itemize}

\item \bbb{Monitoring the variation of \textit{value} per \textit{key}.}

\textcolor{gray}{SELECT} \textit{key, $value^{(key)}_{i}$}

    \textcolor{gray}{FROM} $Data\_Stream$
    
    \textcolor{gray}{WHERE} $value^{(key)}_i-value^{(key)}_{i-1}>Threshold$

    \vspace{0.03in}

% Now we formalize the fourth problem in this paper: 
% (4) 

\end{itemize}

Such tasks monitor the variation in the \textit{value} of items belonging to the same flow. If the difference of \textit{value} between two adjacent items in the same flow exceeds a threshold, we report the flow and the \textit{value} of the latter item.

A scenario corresponding to this task is the real-time detection of \textit{flow gaps} in the data stream \cite{golab2009sequential,szlichta2016effective,golab2010data,golab2003issues,golab2005update}. In this context, the \textit{value} corresponds to the sequence number each item carries, such as the sequence number of frames in videos or the Identification field in an IP header. The sequence number indicates the relative position of an item in a flow, incrementing by one for each new item. Consequently, if we observe a nonconsecutive ($\geq 2$) variation in the sequence number ($seq$) between two adjacent items belonging to the same flow, it suggests the occurrence of item loss or item reordering during the data transmission process. Such flow anomalies are defined as the \textit{flow gap}. If a \textit{flow gap} is too large, it may have a significant negative impact on the Quality of Service (QoS). Therefore, it becomes crucial to report any incident where the variation in $seq$ between two adjacent items within the same flow exceeds a predefined threshold, which we call \textit{major flow gap}.

% {\color{blue}
% One scenario associated with this task involves the real-time detection of \textit{flow gaps} in the data stream. In this context, the \textit{value} corresponds to the sequence number each item carries, such as the sequence number of frames in videos or the Identification field in an IP header. The sequence number indicates the relative position of an item within a flow, incrementing by one for each new item. Consequently, if we observe a nonconsecutive ($\geq 2$) variation in the sequence number ($seq$) between two adjacent items belonging to the same flow, it suggests the occurrence of item loss or item reordering during the data transmission process. Such flow anomalies are defined as a \textit{flow gap}. If a \textit{flow gap} is excessively large, it may have a significant negative impact on the Quality of Service (QoS). Therefore, it becomes crucial to report any incident where the variation in $seq$ between two adjacent items within the same flow exceeds a predefined threshold.
% }

% , which is one of the primary causes of decrease in the quality of service (QoS). 
% If we find an inconsecutive ($\geq 2$) variation of $seq$ between two adjacent items in the same flow, it means some items are missing between them during the data transmission process. If the number of missing items is too large, it may severely impact the QoS. So we need to report it once we find the difference of $seq$ between two adjacent items in the same flow exceeds a threshold.

% \MainIdea{use cases}

% What follows are four typical use cases of detecting \textit{flow gaps}:
\reviewB{While consistent dropping of a few items over a prolonged period is undeniably important, the sudden and substantial loss of data items, as detected by major-flow-gap, can have immediate and severe consequences in certain scenarios. For example:}

% \bbb{Use case 1: Detection of video and audio manipulation.}
% Frame dropping is a common way of video and audio manipulation by deleting consecutive frames from original materials. This kind of manipulation is not detectable by single image forensic techniques, because it only change the temporal aspect of the materials. Our algorithm can efficiently detect the removal of any number of consecutive frames, which satisfies the requirement of solution to this problem.
% The Federated Learning (FL) \cite{shokri2015privacy,mcmahan2017communication,konevcny2016federated, li2020federated, kairouz2021advances} is a rising machine learning framework that trains the model at local clients without exchanging local private data, and aggregates/sums up local gradient vectors in a central Parameter Server.

\bbb{Use case 1: Evaluation of real-time video and audio communication quality\cite{mullin2001new,imtiaz2009performance,al2012survey,watson2001new,watson1996evaluating,jokisch2016audio,rubino2006quantifying,mohamed2002study}.}
A frame drop occurs when one or more frames in a video or audio sequence are lost in the transmission process, which is frequently attributed to network congestion or an unstable internet connection.
\reviewB{While people often neglect or are able to tolerate a certain level of frame loss in video and audio if the number of frames lost at once is small. A one-time drop of many frames can lead to missing a crucial piece of information or significant visual or auditory disruptions that can affect the overall user experience. Therefore, reporting the \textit{major flow gaps} is our primary goal in order to evaluate the communication quality. Besides, it is common that a few frames are lost in the video and audio data stream. It would be burdensome to report and analyze them all.}

% This phenomenon is frequently attributed to network congestion or an unstable internet connection. Our algorithm is adept at efficiently detecting the absence of any number of consecutive frames. This capability empowers it to assess communication quality based on the frame drop situation.

% {\color{blue}
% A frame drop, characterized by the non-display of one or more frames in a video or audio sequence, results in perceptible disruptions. This phenomenon is frequently attributed to network congestion or an unstable internet connection. Our algorithm is adept at efficiently detecting the absence of any number of consecutive frames. This capability empowers it to assess communication quality grounded in the frame drop scenario.
% }

% , which satisfies the requirement of solution to this problem.

% Frame dropping is a common way of video and audio manipulation by deleting consecutive frames from original materials. This kind of manipulation is not detectable by single image forensic techniques, because it only change the temporal aspect of the materials. Our algorithm can efficiently detect the removal of any number of consecutive frames, which satisfies the requirement of solution to this problem.

\bbb{Use case 2: Detection of malicious packet dropping attack in network\cite{hayajneh2009detecting,sen2007distributed,just2003resisting,chaudhary2014design,terence2019novel,shu2014privacy}.}
Malicious packet dropping attacks pose a significant threat to network security as they inhibit the transmission of regular packets, potentially leading to the loss of crucial information.
\reviewB{Every network system has certain resilience towards attacks. So the situation we need to worry about is when an attack cannot be automatically dealt with. A \textit{major flow gap} can indicate a more aggressive form of attack or a change in attack strategy, which requires immediate attention. }
% These attacks control nodes within the network, prompting them to discard specific or all packets. 
Our algorithm provides a viable countermeasure to this threat. It enables continuous monitoring of packet drops by tracking the variation in sequence numbers for each flow.  Furthermore, our algorithm is designed compactly, allowing it to operate effectively even on devices with limited resources. This capability assists network administrators in pinpointing the causes of packet drops and effective actions can be taken to ensure network performance.
Drawing from the use cases outlined above, we delineate two key requirements for a solution designed for the real-time detection of \textit{flow gaps}:
% }

\bbb{R1: Speed.}
Each item should be processed rapidly enough, which necessitates a solution that is not only simple but also ingenious. The reason behind this lies in the fact that data transmission typically occurs at high speed. Furthermore, the processing nodes along the data transmission path are typically constrained by limited computational power, while the volume of the data stream is immense.
% {\color{blue}
% Given the high-speed nature of data transmission, each item should be processed rapidly to maintain pace. This necessitates a solution that is not only simple but also ingenious. Because the processing nodes along the data transmission path are typically constrained by limited computational power, while the volume of the data stream is frequently immense.
% }

\bbb{R2: Accuracy.} 
% The \textit{flow gap} report should be accurate with small memory overhead. False positive report will increase the burden of scrutiny in the network. Failing to detect \textit{flow gaps} may result in omitting some crucial data loss. To achieve \textbf{R1}, the memory cost should be small enough for the data structure to fit in CPU caches.
% {\color{blue}
The \textit{flow gap} reports should be highly accurate with minimal memory overhead. False positive reports would impose additional scrutiny burdens on the network, and failure to detect \textit{flow gaps} may lead to crucial data losses being overlooked. Besides, to fulfill \textbf{R1}, the memory cost should be minimized to enable the data structure to fit within CPU caches.
% }

% \bbb{R3: Compactness.}
% On one hand, storage capacity is limited on the nodes along the data transmission path. On the other hand, to realize \bbb{R2}, we need to fit our algorithm into CPU caches, where memory is small.

\end{sloppypar}

\presub
\subsection{Our Proposed Solution}
\postsub
\label{intro:propose}
% \MainIdea{overall statement}

In this paper, we introduce \algo as a solution for real-time detection of \textit{flow gaps}, achieved by monitoring the variation of sequence numbers within a flow. We design \algo using Sketch, a concise probabilistic data structure that allows for rapid data processing with manageable errors, as a foundation. Our solution realizes both of the design goals: accuracy and speed.
% , which contributes to \bbb{R1} and \bbb{R2}. Its compactness enables it to fit in CPU caches, which contributes to \bbb{R3}. 
The design of our solution rests on two key ideas: the \textit{similarity absorption} technique and the $civilian$-$suspect$ mechanism. Furthermore, we develop two versions of \algo. The first version incorporates only the \textit{similarity absorption} technique, while the second integrates both the \textit{similarity absorption} technique and the $civilian$-$suspect$ mechanism. The first version, known as Speed-Oriented \algo (\algo-SO), excels in terms of speed while maintaining good accuracy. The second version, termed Accuracy-Oriented \algo (\algo-AO), prioritizes accuracy without significantly compromising speed.
Within the context of this paper, we define the number of items within a flow as the flow size, or item frequency. A flow with a large number of items is considered a large flow, while one with a small number of items is considered a small flow.  An abnormal flow is characterized by \textit{flow gaps} caused by item loss or reordering, while flows devoid of such problems are classified as normal flows.

% {\color{blue}
% In this paper, we introduce \algo as a solution for real-time detection of \textit{flow gaps}, achieved by monitoring the variation of sequence numbers within a flow. We design \algo using Sketch as a foundation, a concise probabilistic data structure that allows for data processing with manageable errors. Our solution successfully meets both design goals: accuracy and speed.
% %
% The design of our solution rests on two key ideas: the \textit{Similarity Absorption} technique and the $civilian$-$suspect$ mechanism. Furthermore, we develop two versions of \algo. The first version incorporates only the \textit{Similarity Absorption} technique, while the second integrates both the \textit{Similarity Absorption} technique and the $civilian$-$suspect$ mechanism. The first version, known as Speed-Oriented \algo (\algo-SO), excels in terms of speed while maintaining reasonable accuracy. In contrast, the second version, termed Accuracy-Oriented \algo (\algo-AO), prioritizes accuracy without significantly compromising speed.
% %
% Within the context of this paper, we define the number of items within a flow as the flow size, or item frequency. A flow with a large number of items is considered a large flow, while one with fewer items is considered a small flow. An abnormal flow is characterized by issues such as item loss or reordering, while flows devoid of such problems are classified as normal flows.
% }

% \MainIdea{similarity absorption}

The \textit{similarity absorption} technique leverages the sequence number as an index to achieve efficient matching, which
saves memory and increases speed. Conventionally, to store the information related to a flow, it's necessary to record the flow ID for future matching. \reviewA{However, in \algo, the need to record the flow ID is eliminated, as we utilize the sequence number for the matching task. The \textit{similarity absorption} technique is to select the recorded sequence number closest to the sequence number of the incoming item as the matched sequence number, which corresponds to the smallest degree of anomaly.} This method is chosen based on the premise that normal flows are prevalent in the data stream and severe anomalies are less likely to occur. In practical scenarios, the flow ID typically consists of a 13-byte five-tuple, whereas the sequence number requires no more than 4 bytes (for instance, the Identification field in IP is 2 bytes, and the Packet Number field in QUIC is 4 bytes). Hence, omitting the flow ID leads to considerable memory savings. It also enhances speed since a single field access can accomplish both matching and inspection tasks.

The $civilian$-$suspect$ mechanism proficiently integrates broad monitoring of all flows with detailed scrutiny of suspicious flows. We divide our data structure into two parts: (1) $civilian$ and (2) $suspect$. \reviewA{The $civilian$-$suspect$ mechanism functions as a digital police force within the network. The $civilian$ component operates like a cyber patrol, broadly monitoring all citizens. When this patrol identifies any suspicious activity, the corresponding citizen transitions to a $suspect$ status. Then, a comprehensive inspection process begins, akin to a detailed investigation conducted by a detective, ensuring network safety and integrity.}
% On one hand, we have patrols for an overall inspection of all $civilian$s in the data stream. On the other hand, 
% once our patrol finds a $suspect$,it will go through a meticulous inspection. 
This strategy leverages our limited resources efficiently, performing a broad inspection of all flows, then focusing more detailed investigation on suspicious flows as flagged by the initial examination.
We use two methods to further improve the accuracy of \algo. First, we randomize flows' sequence number and put flows into small buckets. This ensures that the sequence numbers from different flows are well-distributed, and only a limited number of flows share the entire sequence number range. When a bucket becomes saturated, we implement the Least Recently Used (LRU) replacement policy to update the bucket residents.
Secondly, we enhance accuracy by utilizing fingerprints, which are derived by mapping the original long ID to a shorter bit sequence using a hash function. The fingerprint assists the sequence number in the matching process.
% {\color{blue}This ensures that the sequence numbers from different flows are well-distributed, and only a limited number of flows share the entire sequence number range. When a group becomes saturated, we implement the Least Recently Used (LRU) replacement policy to update the group members. 
% Secondly, we enhance accuracy by utilizing fingerprints, which are derived by mapping the original long ID to a shorter bit sequence (fingerprint) using a hash function. The fingerprint assists in the sequence number matching process.
% }

To further boost \algo's speed, we utilize the Single Instruction and Multiple Data (SIMD) operations. In the typical implementation of the LRU replacement policy, storing and comparing timestamps leads to additional memory and time costs. The SIMD operations allow us to eliminate these costs by processing all flows in a bucket in parallel, keeping them in chronological order. Consequently, the least recently used flow is simply the last one in the bucket. This approach obviates the need for storing timestamps and allows us to identify the least recently used flow in $O(1)$ time, as opposed to $O(w)$, where $w$ denotes the number of flows maintained in each bucket. It also reduces the time complexity of matching from $O(w)$ to $O(1)$. Although there is a limit on the number of flows we can operate on in parallel using SIMD, the experiments in Section~\ref{art:exp:para} shows that the optimal bucket size $w$ falls under the limit.
% {\color{blue}
% To further boost \algo's speed, we utilize Single Instruction and Multiple Data (SIMD) operations. In the typical implementation of the LRU replacement policy, storing and comparing timestamps leads to additional memory and time costs. SIMD operations allow us to eliminate these costs by processing all flows in a group simultaneously, keeping them in chronological order. Consequently, the least recently used flow is simply the last one in the group. This approach obviates the need for storing timestamps and allows us to identify the least recently used flow in $O(1)$ time, as opposed to $O(|G|)$. It also reduces the time complexity of matching from $O(|G|)$ to $O(1)$.
% }

% \MainIdea{math and exp}
% The mathematical analysis shows that... The experimental results show that...
We devise a Straw-man solution (Section~\ref{strawman}) for better comparison.
Extensive experiments were performed to evaluate the accuracy and speed of our algorithm. The results show that the accuracy of \algo-SO and \algo-AO achieves about 1.4 and 1.6 times higher than the Straw-man solution, respectively. 
% {\color{blue}The results show that the accuracy of \algo-SO and \algo-AO improve by 40\% and 60\% compared to the Straw-man solution, respectively.} 
\algo-SO is 3 times faster than the Straw-man solution, while \algo-AO is 2.5 times faster than the Straw-man solution.
The experiments also show that \algo is memory-efficient, for it only needs 64 KB to handle 55M items.
% (take \alg-DDSketch on dataset CAIDA for example).
%
We also theoretically prove that \algo achieves a high accuracy with limited memory. 

% {\color{blue}
% We design a Straw-man solution (Section~\ref{strawman}) for improved comparison. 
% Extensive experiments were performed to evaluate the accuracy and speed of our algorithm. The results show that the accuracy of \algo-SO and \algo-AO is about 4 and 6 times higher than the Straw-man solution, respectively. \algo-SO is also 1.6 times faster than the Straw-man solution, while \algo-AO has comparable speed. 
% %
% The experiments also show that \algo is memory-efficient, for it only needs 64 KB to handle 55M items.
% %
% We also theoretically prove that \algo achieves a high accuracy with little memory. 
% }

% \bbb{Key contributions:}

\subsection{Key Contributions:}
\begin{itemize}
    \item To the best of our knowledge, we are the first work that formalizes the task of monitoring the variation in the \textit{value} of items belonging to the same flow, thus addressing a previously unexplored area in the research fields.
    \item We develop the \textit{similarity absorption} technique, designed to save memory and increase speed.
    \item We develop the $civilian$-$suspect$ mechanism, which optimally leverages limited resources to provide both comprehensive monitoring of all flows and detailed scrutiny of suspicious flows.
    \item We implement \algo-SO and \algo-AO on the CPU platform. Compared to the Straw-man solution, both \algo-SO and \algo-AO not only offer superior accuracy but also offer superior speed. All associated codes are publicly accessible on GitHub \cite{link}.
\end{itemize}

    \presec
\section{BACKGROUND AND Related Work}
\postsec
\label{sec:relate}

\reviewC{
To the best of our knowledge, no previous work can be directly applied to accomplish real-time \textit{flow gap} monitoring in data streams.  Though some research has targeted the detection of packet dropping attacks in networks, these studies have different problem definitions and application scenarios compared to our work. 
% Specifically, previous work primarily focuses on malicious attack to the network by manipulating nodes in the packet transmission path. In contrast, our research concentrates on monitoring any loss or reordering of items for any reasons in a per-flow granularity. 
Besides, our work meets both requirements delineated in Section~\ref{intro:background}, \textbf{R1: Speed} and \textbf{R2: Accuracy},  which are not met by previous approaches.
}

\subsection{Statistics-based mechanisms} 
The methodologies presented in \cite{sultana2011provenance,shu2014privacy} rely on historical records to form a correlation or distribution pattern.  If the following packet do not conform the pattern established by their predecessors, the system will report a packet dropping attack. This approach, however, can only detect packet drops after a significant number of packets have been collected and processed, thereby lacking in speed. Additionally, maintaining records of substantial packet information is memory consuming. The need to transmit topology data to a base station for processing also imposes a significant burden on the network. Furthermore, constant network fluctuations necessitate frequent pattern modifications, making the system susceptible to false-positive attack reports due to changing network conditions. These models may also overlook attacks exhibiting changing behaviors.

%Furthermore, the network conditions change all the time. As the result, it has to modify its pattern all the time, putting the system at risk of making false positive attack reports due to changing network condition or omitting attacks with changing behaviors.

\subsection{Trust-based mechanisms}
In \cite{s2012efficient,sanchez2015model,rmayti2017stochastic}, each node is required to monitor its neighboring nodes to gather various status information about the received and forwarded packets. The observed features are then mathematically modeled and used as input to a probability estimate function to determine whether a particular node is malicious. However, this method has notable drawbacks. Firstly, it necessitates the transmission of a substantial number of packets before any packet drops can be detected, and recording this vast amount of information is memory consuming .
% Second, the detection granularity cannot reach per-packet.
Secondly, this method demands the exchange of numerous messages between nodes, which consumes significant bandwidth.

%Secondly, messages need to be exchanged between nodes, which consumes significant bandwidth.

\subsection{ML-based mechanisms}
The approaches adopted by \cite{joseph2010cross,kurosawa2007detecting} incorporate Machine Learning for their detection methodology. They initially gather multidimensional features from the network, MAC, and physical layers by monitoring events and computing topology statistics. These features are then fed into a model to classify an unknown behavior. However, a primary issue with this kind of solution is that it is heavyweight for energy-limited nodes to extract features from packets and do machine learning at the same time. The processing overhead and memory cost are unacceptable.

\presec
\section{Problem Statement}
\postsec
\label{sec:problem}
 
\begin{sloppypar}
\begin{definition}
    \textbf{Data stream.} The data stream is a series of items appearing in sequence. For a given data stream $\mathcal{DS} = \{e_1,e_2,e_3,\ldots,e_i,\ldots\},$ each item $e$ in $\mathcal{DS}$ contains an ID field and a sequence field: $e= \langle FID,SEQ \rangle $. $FID$ stands for the flow to which $e$ belongs, while $SEQ$ stands for its relative position in the flow. In other words, $FID$ serves as an unordered index and $SEQ$ serves as an ordered serial number.
    
    % Items from the same flow share a common $FID$, with $SEQ$ typically incrementing by one from an item to its immediate successor in the same flow. In sum, $FID$ serves as an unordered index and $SEQ$ serves as an ordered serial number.

    % \textcolor{blue}{\textbf{Data Stream Definition.} A data stream is characterized as a sequence of items, each appearing consecutively. For a given data stream, denoted as $\mathcal{DS} = \{e_1,e_2,e_3,\ldots,e_i,\ldots\}$, each item, $e$, is comprised of two distinct fields: an ID field and a sequence field, symbolized as $e= \langle FID,SEQ \rangle $. Here, $FID$ represents the flow to which item $e$ is attributed, while $SEQ$ corresponds to the item's relative position within said flow. All items deriving from the same flow exhibit the same $FID$, with $SEQ$ typically incrementing by one unit from an item to its immediate successor within the same $FID$ context. To summarize, while $FID$ serves as an unordered index, $SEQ$ acts as an ordered serial number.}

\end{definition}    

\reviewC{
\begin{definition}
\textbf{Monitoring flow gaps.}
The \textit{flow gap} describes the nonconsecutive variation ($\geq 2$) in $SEQ$ between adjacent items from the same flow. This irregularity may arise due to item loss or item reordering throughout the data transmission process. The magnitude of a flow gap is indicative of the severity of the issue it represents. 
% A \textit{flow gap} that surpasses a predetermined threshold $\mathcal{T}_1$ should be detected and reported. Moreover, our focus is primarily on the immediate and short-term problems occurring within the data stream. Addressing these issues promptly not only requires fewer resources but also helps prevent minor problems from exacerbating into more significant ones. Therefore, an upper limit, $\mathcal{T}_2$, is set to restrict the range of \textit{flow gaps} we are interested in. 
Specifically, the relationship between the variation in $SEQ$ of two adjacent items and the corresponding situation is delineated as follows:
\end{definition}    
}

\begin{equation} \label{situation}
situation=\left\{
\begin{aligned}
&matched\left\{
\begin{aligned}
neglect, & \quad var \in (-\mathcal{T}_2,1) \\
normal, & \quad var=1 \\
minor \ gap, & \quad var \in [2,\mathcal{T}_1) \\
major \ gap, & \quad var \in [\mathcal{T}_1,\mathcal{T}_2) \\
\end{aligned}
\right.\\
&not \ matched, \quad var \in (-\infty,-\mathcal{T}_2]\cup [\mathcal{T}_2,+\infty) 
\end{aligned}
\right.
\end{equation}

\reviewC{
\bbb{Explanation.} In the above equation, the \textit{major gap} is the condition we aim to detect and report. Our focus is on identifying only positive gaps, because negative gaps are usually accompanied by their positive counterparts. Therefore, we can neglect the variation $\in (-\mathcal{T}_2,1)$. If the situation falls in $\{neglect,\ normal,\ minor \ gap,\ major\ gap\}$, we call it $matched$, because we deem the two items belong to the same flow.
}

\reviewA{
$\mathcal{T}_1$ and $\mathcal{T}_2$ serve as valuable parameters to bound the detection scale of the algorithm, ensuring efficient operation, clear diagnostics, and adaptability across different applications. $\mathcal{T}_1$ is designed to adjust the tolerance level of flow gaps. A \textit{flow gap} that surpasses $\mathcal{T}_1$ should be detected and reported. 
}

\reviewA{
As for $\mathcal{T}_2$, it guarantees controlled response time and limited computational resources. The reason why we set $\mathcal{T}_2$ is similar to why we set a "timeout" threshold in many protocols (\eg, TCP, HTTP, DNS).  In dynamic environments, where data streams are volatile, waiting indefinitely for a response can be detrimental, because we cannot let a possibly broken connection take up the memory and computational resources forever. Besides, as shown in Figure~\ref{fig:datasetDistri:gap}, most of the gaps have small size. Therefore, by capping the detectable gap with $\mathcal{T}_2$, we bring efficiency and predictability to the system.  
% For flow gaps exceeding $\mathcal{T}_2$, usually it indicates a severe interruption in the network and we cannot know when the connection can be restored. So it is reasonable to consider it a new flow when the subsequent items finally come.
% Moreover, our focus is primarily on the immediate and short-term problems occurring within the data stream. Addressing these issues promptly not only requires fewer resources but also helps prevent minor problems from exacerbating into more significant ones. Therefore, an upper limit, $\mathcal{T}_2$ is set to restrict the range of \textit{flow gaps} we are interested in. Thus,
}

\reviewA{Depending on the application, $\mathcal{T}_1$ and $\mathcal{T}_2$ can be adjusted. For high-reliability applications, smaller $\mathcal{T}_1$ and $\mathcal{T}_2$ might be appropriate, forcing a quicker reconnection or system response. For more lenient applications, $\mathcal{T}_1$ and $\mathcal{T}_2$ can be increased, allowing for longer interruptions before deeming the flow undetectable.}

% There are two scenarios of the problem stated above:

% \begin{scenario}
%     (\textbf{Real-time report of major gap}). This scenario requires an immediate report of $major \ gap$ once it happens. So we need to report timely which item in which flow has just encountered a $major \ gap$. This kind of report enables administrators to detect and rectify issues promptly.  
% \end{scenario}

% \presub
% \begin{scenario}
%     (\textbf{Periodic report of major gap}). This scenario requires reporting the statistics information of $major \ gap$ once in a while. With limited resource, over-frequent reports may cause great burden for the devices. Therefore, periodic reports with adequate time interval would suffice.
% \end{scenario}

\end{sloppypar}
    
\begin{table}[!ht]
    \vspace{0in}
\setlength{\abovecaptionskip}{0.1cm}
    \caption{Notations.}
    \begin{tabular}{c|l}
    \hline
    \textbf{Symbol} & \textbf{Meaning}\\
    \hline
    $FID$ & the flow ID of an incoming item\\
    % $FID^c$ & the flow ID recorded in a cell\\
    $SEQ$ & the sequence number of an incoming item\\
    % $SEQ^{\prime}$ & the randomized sequence number of an incoming item\\
    $SEQ^c$ & the sequence number recorded in a cell\\
    $B$ & the bucket array of \algo\\
    $d$ & the number of buckets in \algo\\
    $w$ & the number of cells in a bucket\\
    $c$ & the number of \textit{civilians} in a bucket of \algo-AO\\
    $s$ & the number of \textit{suspects} in a bucket of \algo-AO\\
    $B[i][j]$ & the $SEQ^c$ in the $j^{th}$ cell of the $i^{th}$ bucket\\
    $h(\cdot)$ & a hash function \\
    % $FP$ & the fingerprint of an incoming item\\
    % $FP^c$ & the fingerprint recorded in a cell\\
    % $l_f$ & the number of bits in a fingerprint\\
    \hline
    \end{tabular}
    \label{tab:note}
    \vspace{-0.2in}
\end{table}

\presec
\section{\algo}
\postsec

In this section, we propose our solution in details. The symbols frequently used are shown in table I.

\subsection{Overview}
We devised two versions of \algo: Speed-Oriented \algo (\algo-SO) and Accuracy-Oriented \algo (\algo-AO). \algo-SO is neat and fast, capable of achieving high accuracy on most real-world scenarios. 
% But its accuracy is limited on some extreme scenarios. 
While \algo-AO is a little more complicated, but can achieve robust performance even on extreme scenarios. \algo-SO (Section~\ref{sec:SO}) utilizes the \textit{similarity absorption} technique, which uses the sequence number to serve as the index number to accomplish matching. \algo-AO (Section~\ref{sec:AO}) utilizes both the \textit{similarity absorption} technique and the $civilian$-$suspect$ mechanism. The $civilian$-$suspect$ mechanism proficiently integrates broad monitoring of all flows with detailed scrutiny of suspicious flows. Furthermore, we develop multiple optimization techniques (Section~\ref{sec:OP}) to further improve the performance.

% \textcolor{blue}{
% We have developed two variants of the algorithm: the Speed-Oriented variant (\algo-SO) and the Accuracy-Oriented variant (\algo-AO). \algo-SO, while being streamlined and efficient, is designed to deliver high accuracy under most real-world scenarios, with some limitations under extreme conditions. Conversely, \algo-AO, though slightly more complex, ensures consistent performance across all scenarios, including extreme ones.
% \algo-SO (detailed in Section~\ref{sec:SO}) incorporates the technique of \textit{similarity absorption}, wherein the sequence number doubles as the index number to achieve successful matching. On the other hand, \algo-AO (discussed in Section~\ref{sec:AO}) harnesses both \textit{similarity absorption} and the $civilian$-$suspect$ mechanism. The $civilian$-$suspect$ mechanism serves as an effective tool for merging high-level surveillance of overall flow with detailed monitoring of suspicious flows.
% Additionally, we've integrated multiple optimization strategies (outlined in Section~\ref{sec:OP}) to further enhance the overall performance of the algorithms.
% }

\subsection{Proposed Speed-Oriented solution}
\label{sec:SO}
\begin{figure}[!htbp]
    \centering
    \includegraphics[width=0.98\linewidth]{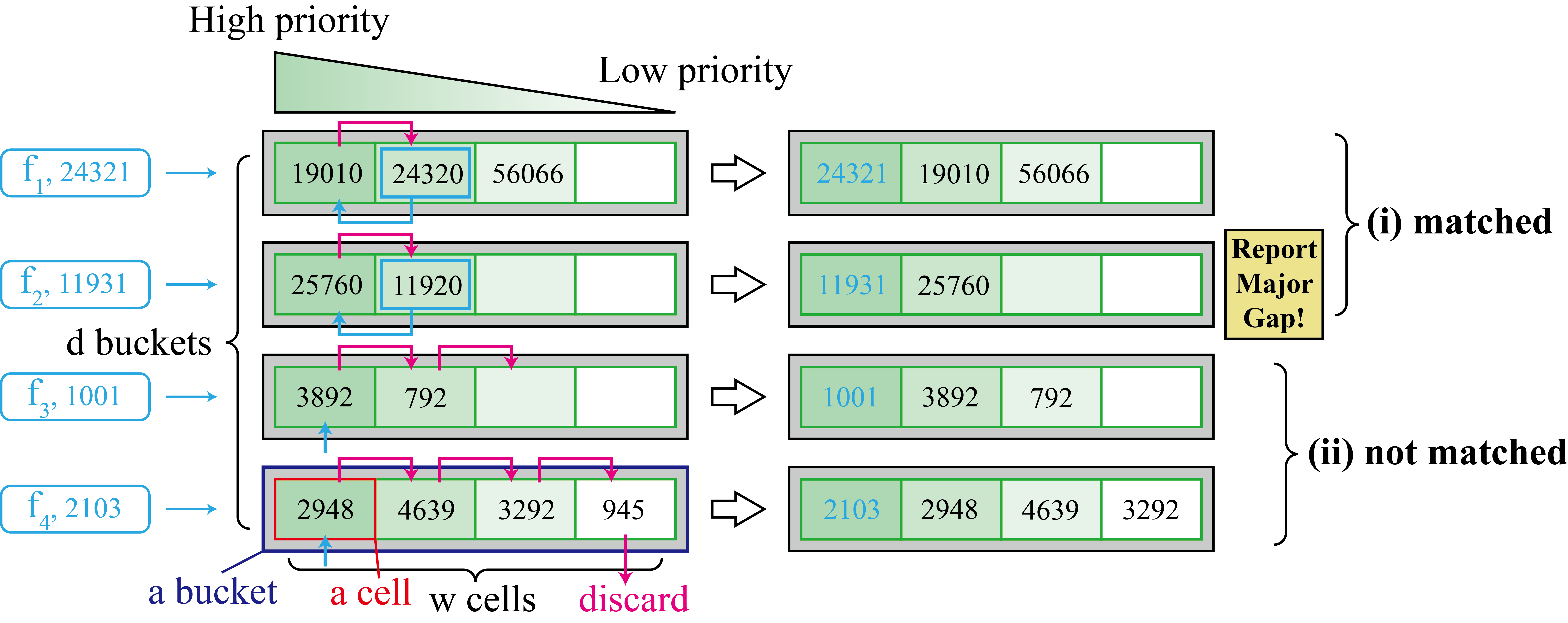}
    \caption{{ Illustration of \algo{}-SO.} This is a \algo{}-SO with $d=4$ buckets and each bucket has $w=4$ cells. $e_1$, $e_2$, $e_3$, $e_4$ arrive in order and are mapped to different buckets. The final result after the arrival of these four items lies on the right.
    }
    \label{fig:algo:SO}
    % \postcap
    \vspace{-0.15in}
\end{figure}

% \begin{algorithm}[ht]
%     \caption{Speed-Oriented solution}
%     \label{algo:SO}
    
%     \DontPrintSemicolon
    
%     \KwInput{Item : $\langle FID, SEQ \rangle $, the bucket array : $B$}
%     \KwOutput{$True$ if it is a major gap}
%     \BlankLine
%     $idx = h(FID)\%d$\\
%     %$seq\_offset = seq\_hash(FID)$ \\
    
%     %$SEQ' = (SEQ + seq\_offset)$ mod $MAX\_SEQ$\\

%     $match$ = the $j$ minimize \{$abs(SEQ - B[idx][j])$\}\\
%     $var = SEQ - B[idx][match]$
    
%     \If {$-\mathcal{T}_2 < var < \mathcal{T}_1$}
%     {
%        \textcolor{gray}{/* neglect, normal, minor gap */}\\
%         $report = False$\\
%         $B[idx][match] = max(SEQ, B[idx][match])$\\
%         $B[idx].rearrange$\footnotemark $(match,0)$
%     }
%     \textbf{else} \If{$\mathcal{T}_1 \leq var < \mathcal{T}_2$}
%     {
%             \textcolor{gray}{/* major gap */}\\
%             $report = True$\\
%         $B[idx][match] = max(SEQ, B[idx][match])$\\
%         $B[idx].rearrange$ $(match,0)$
%     }
%     \Else{
%             \textcolor{gray}{/* not matched */}\\
%             $report = False$\\
%         $B[idx][w-1] = SEQ$\\
%         $B[idx].rearrange$ $(w-1,0)$
%     }
    
%     \Return $report$

% \end{algorithm}
% \input{tex/fig/SO-pseudo}

\subsubsection{Data structure}~

As shown in figure~\ref{fig:algo:SO}, the data structure of \algo-SO is composed of a bucket array (denoted as $B$) of $d$ buckets, and a hash function $h(\cdot)$. There are $w$ cells in one bucket, each bearing a field of sequence number $SEQ^c$ in it. The $SEQ^c$ in the $j^{th}$ cell of the $i^{th}$ bucket is denoted as $B[i][j]$.
% When a new item $e=<FID,SEQ>$ arrives, it is mapped into $B[h(FID)\%d]$.
% Besides, $b(\cdot)$ is for randomizing the sequence number of every flow.
% and $\mathcal{F}(\cdot)$ is for calculating the fingerprint of every flow. 

% This data structure realizes the idea of grouping, where each bucket is a group and $w$ is the group size. The smaller a group is, the less risk that two flows' sequence number overlap (we call it $seq$ collision).
% % for the flows in the group. 
% But $w$ cannot be too small, because it will increases the risk of more than $w$ flows competing for the $w$ cells in a bucket, as what experiment results show in Section~\ref{art:exp:para}.

\subsubsection{Monitoring operations}~
% \footnotetext{$bucket.rearrange(a,b)$ means that removing the $a$-th cell from $bucket$ and reinsert it into the $b$-th position, specifically, if $0\le b<a\le w-1$, it is equivalent to executing the following statement: $tmp = bucket[a]$; $bucket[a]=bucket[a-1]$; $bucket[a-1]=bucket[a-2]$; $\cdots$; $bucket[b+1]=bucket[b]$; $bucket[b]=tmp$}
\begin{sloppypar}

When a new item $e=\langle FID,SEQ \rangle $ arrives, we map it to the bucket $B[h(FID)\%d]$ and choose the cell with the $SEQ^c$ closest to $SEQ$.
% Try to find the largest $B[h(FID)\%d][j] \ (j\in\{1,2,\ldots,w\})$ $s.t.$ $SEQ^{\prime}>B[h(FID)\%d][j]$. 
% In other words, our choice is the sequence number that indicates the least number of item loss. 
% The reason why we choose in this way is that normal flows are in the majority, and network is functioning well most of the time. The probability of dropping more than $n$ items decreases as $n$ grows. Therefore, the situation that loses least number of items is the most likely to happen. 
Suppose the $j^{th}$ cell is the one. We plug $var=(SEQ-B[h(FID)\%d][j])$ into equation (\ref{situation}) to get the situation:
% \footnotetext{$bucket.rearrange(a,b)$ means that removing the $a$-th cell from $bucket$ and reinsert it into the $b$-th position, specifically, if $0\le b<a\le w-1$, it is equivalent to executing the following statement: $tmp = bucket[a]$; $bucket[a]=bucket[a-1]$; $bucket[a-1]=bucket[a-2]$; $\cdots$; $bucket[b+1]=bucket[b]$; $bucket[b]=tmp$}
\begin{itemize}[leftmargin=*]
    \item \textbf{(i)} If the situation is $matched$, we call the chosen cell the matched cell. We set $B[h(FID)\%d][j]=SEQ$ if $SEQ$ is larger. If the situation is $major\ gap$, we report it.

    \item \textbf{(ii)} If the situation is $not\ matched$, meaning none of the $SEQ^c$ recorded in this bucket matches $e$, we will find a cell for $e$. If there exists an empty cell, we will initialize it by setting its $SEQ^c=SEQ$. If there is no empty cell in the bucket, we will empty the cell containing the least recently used (LRU) flow and initialize it by setting its $SEQ^c=SEQ$.

\end{itemize}

% \footnotetext{$bucket.rearrange(a,b)$ means that removing the $a$-th cell from $bucket$ and reinsert it into the $b$-th position, specifically, if $0\le b<a\le w-1$, it is equivalent to executing the following statement: $tmp = bucket[a]$; $bucket[a]=bucket[a-1]$; $bucket[a-1]=bucket[a-2]$; $\cdots$; $bucket[b+1]=bucket[b]$; $bucket[b]=tmp$}

Traditionally, two primary methods are employed to implement the LRU (Least Recently Used) replacement policy. The first method is to store and compare timestamps, which leads to extra memory and time cost. The second method is to arrange the items chronologically, obviating the need for extra memory but necessitating a constant rearrangement of items each time a new one arrives. We choose the second method in our solution. To optimize the \textit{item rearrangement} process, we utilize the tool of Single Instruction and Multiple Data (SIMD) to decrease the time complexity of rearranging to $O(1)$. The detailed process of implementing LRU through SIMD is elaborated upon in Section~\ref{OP:SIMD}. The \textit{item rearrangement} operation proceeds as follows: the priority of cells within a bucket decreases sequentially from the first to the last. Every time a new item arrives, we put it in the first cell of the bucket while existing items are shifted backward. When we need to empty a cell, we simply empty the last cell because it holds the item of the lowest priority (\textit{i.e.,} the least recently used). 

The corresponding pseudo-code is shown in Algorithm~\ref{algo:SO}.
\begin{algorithm}[ht]
    \caption{Speed-Oriented solution}
    \label{algo:SO}
    
    \DontPrintSemicolon
    
    \KwInput{Item : $\langle FID, SEQ \rangle $, the bucket array : $B$}
    \KwOutput{$True$ if it is a major gap}
    \BlankLine
    $idx = h(FID)\%d$\\
    %$seq\_offset = seq\_hash(FID)$ \\
    
    %$SEQ' = (SEQ + seq\_offset)$ mod $MAX\_SEQ$\\

    $match$ = the $j$ minimize \{$abs(SEQ - B[idx][j])$\}\\
    $var = SEQ - B[idx][match]$
    
    \If {$-\mathcal{T}_2 < var < \mathcal{T}_1$}
    {
       \textcolor{gray}{/* neglect, normal, minor gap */}\\
        $report = False$\\
        $B[idx][match] = max(SEQ, B[idx][match])$\\
        $B[idx].rearrange$\footnotemark $(match,0)$
    }
    \textbf{else} \If{$\mathcal{T}_1 \leq var < \mathcal{T}_2$}
    {
            \textcolor{gray}{/* major gap */}\\
            $report = True$\\
        $B[idx][match] = max(SEQ, B[idx][match])$\\
        $B[idx].rearrange$ $(match,0)$
    }
    \Else{
            \textcolor{gray}{/* not matched */}\\
            $report = False$\\
        $B[idx][w-1] = SEQ$\\
        $B[idx].rearrange$ $(w-1,0)$
    }
    
    \Return $report$

\end{algorithm}

\footnotetext{$bucket.rearrange(a,b)$ means that removing the $a$-th cell from $bucket$ and reinsert it into the $b$-th position, specifically, if $0\le b<a\le w-1$, it is equivalent to executing the following statement: $tmp = bucket[a]$; $bucket[a]=bucket[a-1]$; $bucket[a-1]=bucket[a-2]$; $\cdots$; $bucket[b+1]=bucket[b]$; $bucket[b]=tmp$}

% \footnotetext{$bucket.rearrange(a,b)$ means that removing the $a$-th cell from $bucket$ and reinsert it into the $b$-th position, specifically, if $0\le b<a\le w-1$, it is equivalent to executing the following statement: $tmp = bucket[a]$; $bucket[a]=bucket[a-1]$; $bucket[a-1]=bucket[a-2]$; $\cdots$; $bucket[b+1]=bucket[b]$; $bucket[b]=tmp$}

\end{sloppypar}

\subsubsection{Example.}~

In the example shown in Figure~\ref{fig:algo:SO}, we set $\mathcal{T}_1=5$ and $\mathcal{T}_2=30$. 
% For the simplicity of expression, the sequence number in Figure~\ref{fig:algo:SO} is the randomized sequence number.
% As shown in Figure~\ref{fig:algo:SO}, 

\begin{itemize}[leftmargin=*]

\item When $e_1$=$\langle FID=f_1,SEQ=23421 \rangle $ arrives, it successfully finds a matched cell with $SEQ^c=23420$ in the bucket. Since $23421-23420=1$, the situation is considered \textit{normal}, Consequently, we update $SEQ^c$ to 23421 and rearrange the items in the bucket.

\item When $e_2=\langle FID=f_2,SEQ=11931 \rangle $ arrives, it successfully finds a matched cell with $SEQ^c=11920$ in the bucket. Since $\mathcal{T}_1 \leq 11931-11920<\mathcal{T}_2$, a $major\ gap$ is found and reported. Then, we update the $SEQ^c$ to 11931 and rearrange the items in the bucket. 

\item When $e_3=\langle FID=f_3,SEQ=1001 \rangle $ arrives, it fails to find a matched cell, prompting the insertion of $e_3$ into an empty cell in the bucket. The items in the bucket are subsequently rearranged.

\item When $e_4=\langle FID=f_4,SEQ=2103 \rangle $ arrives, it fails to find a matched cell or an empty cell. So it empties the last cell and inserts itself into the first cell. Finally, all the items in the bucket are rearranged.
% Note that in our implementation, the update/insertion and the rearrangement are accomplished within one operation respectively

\end{itemize}

\subsubsection{Analysis.} \label{SO:analysis}~

\bbb{Similarity Absorption.} In our design, the sequence number also serves as the index number to accomplish matching. The traditional way that uses the ID to do matching is both memory consuming and time consuming. In our solution, we use a greedy algorithm that selects the recorded sequence number closest to the sequence number of the incoming item as the matched sequence number.
This is because in most cases, we expect the
% the sequence number of different flows are scattered and
network to function properly. 
% The probability of a \textit{flow gap} occurring decreases as its length $n$ grows. 
The recorded sequence number that indicates the smallest \textit{flow gap} is most likely to be the matched sequence number.
% Furthermore, we employ SIMD to increase the speed to a time complexity of $O(1)$.

\bbb{Grouping.} To reduce the risk of sequence number collision ($seq$ collision), we conduct flow grouping in our approach. Each bucket in \algo-SO represents a group, with a group size of $w$. By grouping flows, we can assure that only a limited number of flows share the entire sequence number range.
% reduce the chances of sequence number overlap between two flows.
Thus, A smaller group size implies a lower risk of $seq$ collision.
% , as each flow has a larger average range of sequence numbers. 
However, it is important to avoid setting $w$ too small. If $w$ is too small, it increases the likelihood of more than $w$ flows competing for the limited $w$ cells in a bucket. The experiment results related to this issue are shown in Section~\ref{art:exp:para}.

% This issue and its associated risks are discussed further in Section~\ref{art:exp:para}, where experimental results are presented.

% \input{Figures/front/rules_AO}

% \vspace{1in}
\subsection{Proposed Accuracy-Oriented solution}
\label{sec:AO}
\begin{figure}[!htbp]
    \centering
    \includegraphics[width=0.98\linewidth]{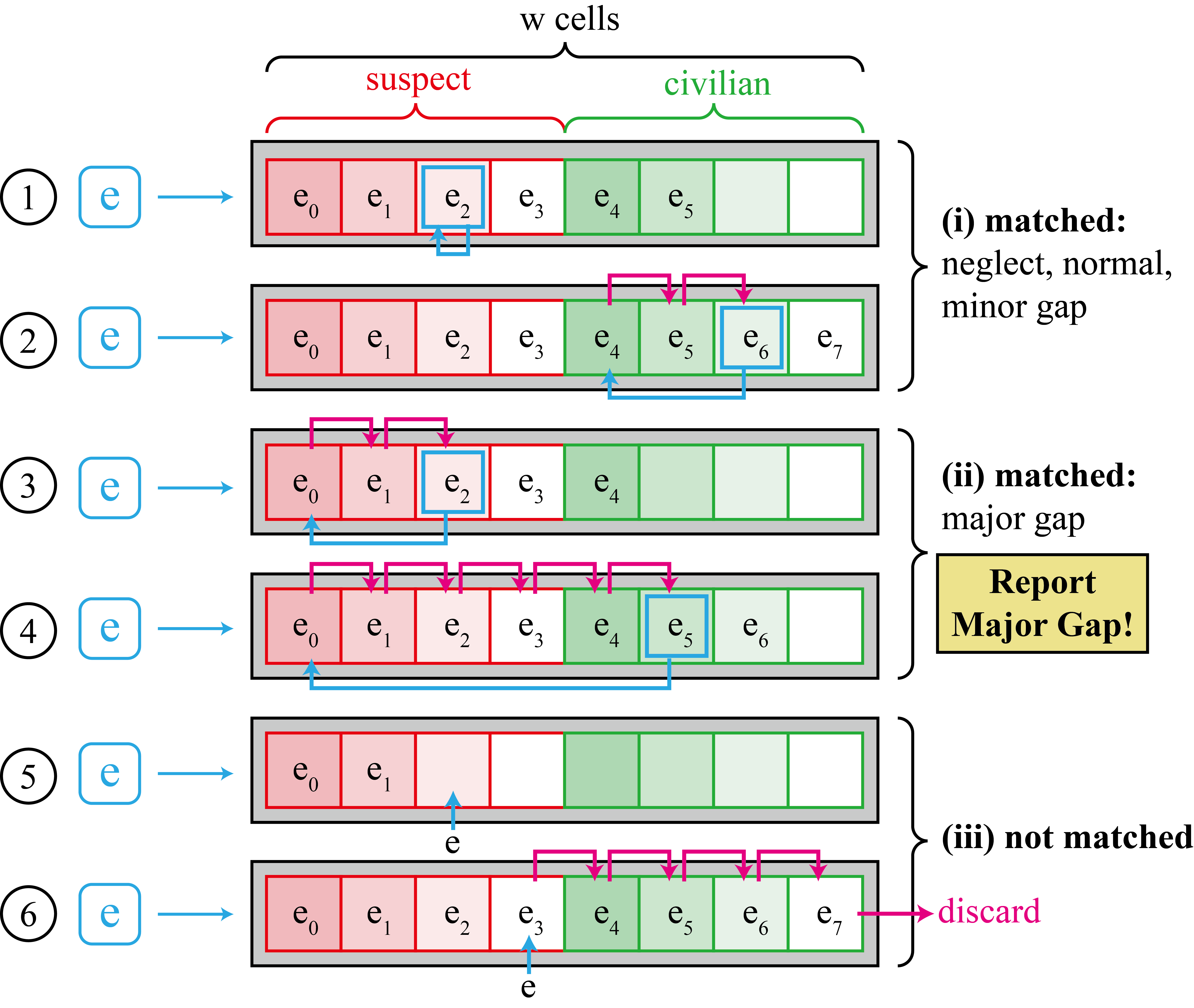}
    \caption{ Monitoring operations in \algo-AO. }
    \label{fig:algo:rules}
    % \postcap
    \vspace{-0.15in}
\end{figure}

% \begin{algorithm}[ht]
%     \caption{Accuracy-Oriented solution}
%     \label{algo:AO}
    
%     \DontPrintSemicolon
    
%     \KwInput{Item : $\langle FID, SEQ \rangle $, the bucket array : $B$}
%     \KwOutput{$True$ if it is a major gap}
%     \BlankLine
%     $idx = h(FID)\%d$\\
%     %$seq\_offset = seq\_hash(FID)$ \\
%     %$fp = fp\_hash(FID)$\\
    
%     %$SEQ' = (SEQ + seq\_offset)$ mod $MAX\_SEQ$\\
%     %$New\_Cell = \langle SEQ, fp \rangle $\\
%     $match$ = the $j$ minimize \{$abs(SEQ - B[idx][j])$\}\\
%     $var = SEQ - B[idx][match]$
    
%     \If {$-\mathcal{T}_2 \leq var \leq \mathcal{T}_1$}
%     {
%         \textcolor{gray}{/* neglect, normal, minor gap */}\\
%         $report$ = $False$ \\
%         $B[idx][match] = max(SEQ, B[idx][match])$\\
%         \If {$match \geq s$}
%         {
%             $B[idx].rearrange(match , s)$
%         }}
%     \textbf{else} \If{$\mathcal{T}_1 < var \leq \mathcal{T}_2$}
%     {
%         \textcolor{gray}{/* major gap */}\\
%         $report$ = $True$\\
%         $B[idx][match] = max(SEQ, B[idx][match])$\\
%         $B[idx].rearrange(match , 0)$
%     }
%     \Else{
%         \textcolor{gray}{/* not matched */}\\
%         $report$ = $False$\\
%         $empty\_cell$ = the index of first empty cell of $B[idx]$\\
%         \If {$empty\_cell$ is not None and $empty\_cell < s$}
%         {
%             $B[idx][empty\_cell] = SEQ$\\
%         }
%         \Else
%         {
%             $B[idx][w-1] = SEQ$\\
%             $B[idx].rearrange$ $(w - 1 , s-1)$\\
%         }
%     }
%     \Return $report$

% \end{algorithm}

\subsubsection{Data structure}~

Similar to \algo-SO, the data structure of \algo-AO consists of a bucket array ($B$) with $d$ buckets and a hash function $h(\cdot)$. \algo-AO differs from \algo-SO in that each bucket in \algo-AO is divided into two parts: (1) $c$ civilians and (2) $s$ suspects ($c+s=w$). Each cell in a bucket contains a field $SEQ^c$. The priority of cells in the civilian part is updated based on the Least Recently Used (LRU) replacement policy. While in the $suspect$ part, we devise a different replacement policy called Least Recently Disrupted (LRD). Instead of assigning a flow with the highest priority\footnote{rearrange the items in a bucket according to the time order} when a new item of it arrives as what $civilian$ does, $suspect$ only assigns the flow with the highest priority when a new $major \ gap$ occurs in it. Every time we need to evict a flow, we evict the one that has not encountered $major \ gap$ for the longest time.
The $civilian$ part of \algo-AO functions as an overall monitor for all the flows, while the $suspect$ part focuses on keeping suspicious flows under constant and meticulous surveillance.

% \begin{algorithm}[ht]
%     \caption{Accuracy-Oriented solution}
%     \label{algo:AO}
    
%     \DontPrintSemicolon
    
%     \KwInput{Item : $\langle FID, SEQ \rangle $, the bucket array : $B$}
%     \KwOutput{$True$ if it is a major gap}
%     \BlankLine
%     $idx = h(FID)\%d$\\
%     %$seq\_offset = seq\_hash(FID)$ \\
%     %$fp = fp\_hash(FID)$\\
    
%     %$SEQ' = (SEQ + seq\_offset)$ mod $MAX\_SEQ$\\
%     %$New\_Cell = \langle SEQ, fp \rangle $\\
%     $match$ = the $j$ minimize \{$abs(SEQ - B[idx][j])$\}\\
%     $var = SEQ - B[idx][match]$
    
%     \If {$-\mathcal{T}_2 < var < \mathcal{T}_1$}
%     {
%         \textcolor{gray}{/* neglect, normal, minor gap */}\\
%         $report$ = $False$ \\
%         $B[idx][match] = max(SEQ, B[idx][match])$\\
%         \If {$match \geq s$}
%         {
%             $B[idx].rearrange(match , s)$
%         }}
%     \textbf{else} \If{$\mathcal{T}_1 \leq var < \mathcal{T}_2$}
%     {
%         \textcolor{gray}{/* major gap */}\\
%         $report$ = $True$\\
%         $B[idx][match] = max(SEQ, B[idx][match])$\\
%         $B[idx].rearrange(match , 0)$
%     }
%     \Else{
%         \textcolor{gray}{/* not matched */}\\
%         $report$ = $False$\\
%         $empty\_cell$ = the index of first empty cell of $B[idx]$\\
%         \If {$empty\_cell$ is not None and $empty\_cell < s$}
%         {
%             $B[idx][empty\_cell] = SEQ$\\
%         }
%         \Else
%         {
%             $B[idx][w-1] = SEQ$\\
%             $B[idx].rearrange$ $(w - 1 , s-1)$\\
%         }
%     }
%     \Return $report$

% \end{algorithm}
\begin{algorithm}[ht]
    \caption{Accuracy-Oriented solution}
    \label{algo:AO}
    
    \DontPrintSemicolon
    
    \KwInput{Item : $\langle FID, SEQ \rangle $, the bucket array : $B$}
    \KwOutput{$True$ if it is a major gap}
    \BlankLine
    $idx = h(FID)\%d$\\
    %$seq\_offset = seq\_hash(FID)$ \\
    %$fp = fp\_hash(FID)$\\
    
    %$SEQ' = (SEQ + seq\_offset)$ mod $MAX\_SEQ$\\
    %$New\_Cell = \langle SEQ, fp \rangle $\\
    $match$ = the $j$ minimize \{$abs(SEQ - B[idx][j])$\}\\
    $var = SEQ - B[idx][match]$
    
    \If {$-\mathcal{T}_2 < var < \mathcal{T}_1$}
    {
        \textcolor{gray}{/* neglect, normal, minor gap */}\\
        $report$ = $False$ \\
        $B[idx][match] = max(SEQ, B[idx][match])$\\
        \If {$match \geq s$}
        {
            $B[idx].rearrange(match , s)$
        }}
    \textbf{else} \If{$\mathcal{T}_1 \leq var < \mathcal{T}_2$}
    {
        \textcolor{gray}{/* major gap */}\\
        $report$ = $True$\\
        $B[idx][match] = max(SEQ, B[idx][match])$\\
        $B[idx].rearrange(match , 0)$
    }
    \Else{
        \textcolor{gray}{/* not matched */}\\
        $report$ = $False$\\
        $empty\_cell$ = the index of first empty cell of $B[idx]$\\
        \If {$empty\_cell$ is not None and $empty\_cell < s$}
        {
            $B[idx][empty\_cell] = SEQ$\\
        }
        \Else
        {
            $B[idx][w-1] = SEQ$\\
            $B[idx].rearrange$ $(w - 1 , s-1)$\\
        }
    }
    \Return $report$

\end{algorithm}

\subsubsection{Monitor operation}~
\label{sec:AO:operation}
\begin{sloppypar}

When an item $e=\langle FID,SEQ \rangle$ arrives, we map it to the bucket $B[h(FID)\%d]$ and select the cell with the $SEQ^c$ closest to $SEQ$. Suppose the $j^{th}$ cell is the one selected. We plug $var=(SEQ-B[h(FID)\%d][j])$ into equation (\ref{situation}) to get the situation.
Figure~\ref{fig:algo:rules} shows the different cases of insertion. The $e_i(i\in \{0,1,2,3,4,5,6,7\})$ in the figure represents the item recorded in the $i^{th}$ cell.

% As shown in Figure\ref{fig:algo:rules}, there are three cases:
\begin{itemize}[leftmargin=*]
    \item \textbf{Case (i)} The matched cell indicates one of the following situations: $neglect$, $normal$ or $minor \ gap$. If the matched cell is in the $suspect$ part (\ding{172}), we simply update the $SEQ^c$ without rearranging items. However, if the matched cell is in $civilian$ part, we not only update the $SEQ^c$ but also rearrange the matched cell to the first cell in the $civilian$ (\ding{173}).

    \item \textbf{Case (ii)} The matched cell indicates a $major \ gap$. First, we report the $major \ gap$. Next, no matter whether the matched cell is in $civilian$ (\ding{174}) or $suspect$ (\ding{175}), we update the $SEQ^c$ and rearrange it to the first cell in $suspect$. If the matched cell is in $civilian$, then $suspect$ must be full before $e$ arrives (we will explain \reviewC{later}). Therefore, the rearrange operation will force $suspect$ to transfer its last item to $civilian$.

    \item \textbf{Case (iii)} There is no matched cell in the bucket. In this situation, we fail to find the information about this flow, so we insert it to $suspect$ for further examination. Because we do not know whether it is an abnormal flow, we assign it the lowest priority when inserted to $suspect$. Specifically, if $suspect$ still has empty cells, we insert $e$ to the first (leftmost) empty cell (\ding{176}). If $suspect$ is full, we transfer the last (rightmost) item in $suspect$ to $civilian$, and insert $e$ to the last cell in $suspect$ (\ding{177}).

\end{itemize}

Using the above operations (\ding{172} to \ding{177}), flows may transition back and forth between the $civilian$ and $suspect$. A flow that becomes suspicious in $civilian$ will be transferred to $suspect$. A flow in $suspect$ that gets squeezed out by other more suspicious flows, which we call "exonerated", will be transferred to $civilian$. When a flow gets squeezed out by other flows in $civilian$, it will simply be discarded. This $civilian$-$suspect$ mechanism proficiently integrates broad monitoring of all flows with detailed scrutiny of suspicious flows. 

The corresponding pseudo-code is shown in Algorithm~\ref{algo:AO}. 
Further details on how these operations sufficiently meet our requirements will be explained in the following paragraphs.

At the beginning of the monitoring process, the data structure is empty and every flow is considered new due to the absence of historical data. We should insert all the flows to $suspect$ as priority, because space is abundant and we can give them as meticulous surveillance as possible. Therefore, when there are empty cells in $suspect$, it implies that the $civilian$ is entirely vacant. Conversely, when $civilian$ contains items, it indicates that the $suspect$ is at full capacity. This need is met by operation \ding{176}.

During the mid-phase of the monitoring process, for a large abnormal flow, if it is at $suspect$, operation \ding{174} will ensure its information is well stored, so we can detect $major\ gaps$ in time. If it is at $civilian$, operation \ding{173} will ensure its information is well stored. Once a $major\ gap$ happens, the flow will be transitioned to $suspect$ via operation \ding{175}. 
For a small abnormal flow, if it is at $suspect$, operation \ding{174} and \ding{172} prevent it from being displaced by normal large flows, thus guaranteeing its data is well stored for $major\ gap$ detection. If it is at $civilian$, it may get squeezed out by large flows. Nevertheless, if that happens, operation \ding{177} will send it to $suspect$, where its data will be well protected, enabling the detection of $major\ gaps$.

% for a normal large flow, if it is at $civilian$, $civilian$ should be able to store its information because of operation \ding{173}. If it is at $suspect$, its priority will not be increased because of operation \ding{172}, and will only be decreased because of operation  \ding{174} and  \ding{175}. Thus, it will be eventually evicted to $civilian$.

\end{sloppypar}

\subsubsection{\reviewC{Discussion}} \label{AO:analysis}~

\bbb{The motivation for \textit{suspect}.} 
% The \algo-AO fulfills both the similarity absorption technique and the $civilian-suspect$ mechanism. $Civilian$-$suspect$ mechanism efficiently and organically combines a rough monitoring on the overall flows and a meticulous monitoring on the suspicious flows. The $civilian$ part functions similarly to \algo-SO. 
The Least Recently Used (LRU) replacement policy implemented in \algo-SO inherently favors large flows. This favoritism is generally justified as large flows often represent crucial or high-priority data\footnote{Large flows frequently signify important data, and by giving these flows precedence, administrators can enhance the system's performance by ensuring the efficient delivery of key information.}. However, this LRU replacement policy does present a potential area for improvement. Specifically, there is a risk that smaller abnormal flows may be displaced by larger normal flows, in which case we lose the information we want and keep the information we do not need. The Least Recently Disrupted (LRD) replacing policy, which is used in the $suspect$ part, guarantees that 
% Instead of updating all items' priority\footnote{rearrange the items in a bucket according to the time order} when a new item arrives as what $civilian$ does, $suspect$ only updates all items' priority when a new $major \ gap$ occurs. Every time we need to evict an item, we evict the one that has not encountered $major \ gap$ for the longest time. 
no matter how large a normal flow is, it crowds out no small abnormal flow. Thus, all the information we want is well protected.

% \bbb{Analysis I:}
% First, we explain why we insert the item to $suspect$ if the situation is "not matched" in case (iii). There are three circumstances when a flow is missing in a bucket: 

% (1) The arriving item is the first one in the missing flow. 

% (2) The missing flow used to be recorded in $civilian$, then it encountered a $major \ gap$ and got kicked out because its items stop coming for a long time. 

% (3) The missing flow used to be recorded in $civilian$. It shared the memory with some large flows and got squeezed out by them. 

% The circumstance we look for is (2). Circumstance (1) is a comparatively rare case which does not impact our performance, because each flow has multiple items other than the first item. To distinguish between (2) and (3), it is worthwhile to insert the item to $suspect$ to get a closer look into it, because in $suspect$, a flow will not be squeezed out by normal flows of any size. Hence, if the circumstance is (2), the missing flow will survive in $suspect$, and we will have a better chance of maintaining the $SEQ^c$ information that would get lost in $candidate$. If the circumstance is (3), the inserted item will soon get exonerated (because $major\ gap$s do not happen and it is at the last cell in $suspect$) and transferred to $civilian$. Therefore, it is efficient at solving our problem to insert an item to $suspect$ if we cannot find a matched cell in $civilian$.

\bbb{The reason to keep \textit{civilian}.}
% Second, we explain why we do not get rid of $civilian$ and just use $suspect$ to undertake the whole monitoring task. 
$Suspect$ is supposed to protect the flows that are encountering $major\ gaps$. However, if we do not know any information about the new flows coming into $suspect$, we may have let normal flows squeeze out abnormal flows that are already in $suspect$.
% The flows in $suspect$ can be considered to carry a privilege "interrupted". Without $civilian$, we cannot know if a new flow trying to come into $suspect$ has the privilege, so we cannot give it the highest priority. If we just insert it to the cell with the lowest priority, there will be too many flows crushing in the last cell. Even if there are abnormal flows among them, they will not get a chance to encounter a gap before it gets squeezed out by other flows. 
This is where the $Civilian$ becomes particularly valuable. It offers every flow the opportunity to have its \textit{flow gap} detected, acting as a crucial preliminary screening stage for potentially suspicious flows.
% Finding the suspicious flows can only be accomplished by $civilian$. 
That is why we claim the $civilian$-$suspect$ mechanism proficiently integrates broad monitoring of all flows with detailed scrutiny of suspicious flows.

\subsection{Optimizations}
\label{sec:OP}

\subsubsection{SIMD Acceleration.} \label{OP:SIMD}~

Normally, to find a matched sequence number in a bucket, the time complexity is related to the bucket size $w$. Besides, it requires storing and comparing timestamps to realize LRU replacing policy, which leads to extra memory and time cost. Nevertheless, by using SIMD, we can operate on all fields in a bucket at once with only one instruction, reducing the time complexity of finding a matched sequence number to $O(1)$. As for the implementation of LRU, instead of storing and comparing timestamps, we sort all flows in a bucket according to the time order. Every time an item $e=\langle FID,SEQ \rangle $ arrives, we put the flow $FID$ to the first cell of the bucket and move other flows backward. 

% The implemention details of SIMD are shown in Appendix~\ref{sec:app:simd}.

Specifically, we can fetch all $w$ cells at once using the function $\_mm\_set1\_pi16()$, and use function $\_mm\_subs\_pi16()$ and $\_mm\_min\_pi16()$ to compare $SEQ$ with $SEQ^c$ in all the cells to try to find the matched cell. After updating the $SEQ^c$, we use the function $\_mm\_shuffle\_pi16()$ to rearrange all the flows according to the time order. Every operation above \reviewC{is implemented in C++ and} can be completed in one SIMD instruction ($i.e.$, one CPU cycle). For an example of the rearrange operation, suppose the $j^{th}$ item is the one we want to update, we move the $j^{th}$ item to the $0^{th}$ cell, the $0^{th}$ item to the $1^{st}$ cell, the $1^{st}$ item to the $2^{nd}$ cell, the $2^{nd}$ item to the $3^{rd}$ cell, ..., the $(j-1)^{th}$ item to the $j^{th}$ cell. The $(j+1)^{th}$ to $(w-1)^{th}$ items remain at their original cells. The
% the first cell to possess the highest priority and the last cell to possess the lowest priority. The 
priority of cells is defined to diminish from first (leftmost) to last (rightmost). In this way, each time we insert an item $e=\langle FID,SEQ \rangle $, we update flow $FID$'s priority to the highest. When we need to empty a cell, we simply empty the last cell because it has the lowest priority (least recently used). When a bucket is not full, the empty cells will be all at the back of the bucket.

\vspace{0.3in}
\subsubsection{Sequence number randomizing} \label{OP:random}~

To mitigate $seq$ collisions caused by internal problems, such as the possibility of multiple flows starting with the same sequence number and thereby increasing the probability of sequence number collision, we preprocess the sequence number before it enters the \algo algorithm. This preprocessing involves the utilization of a hash function $b(\cdot)$ to generate a random bias $b(FID)$, which is then added to the original sequence number of every item in the same flow.
By adding this random bias to the sequence number, we achieve a desirable outcome: the range of the randomized sequence numbers for every flow is separate
from each other. Consequently, the risk of $seq$ collisions is significantly reduced.

% Experiments regarding the effects of Sequence number randomizing are shown in Appendix~\ref{exp:opt}.
% \vspace{0.3in}
\subsubsection{Fingerprint}~
\label{OP:FP}

In some extreme scenarios, some assistance on matching may further improve the accuracy. We use a hash function $\mathcal{P}(\cdot)$ to generate a fingerprint $\mathcal{P}(FID)$, which is a short bit sequence. We store the fingerprint with the sequence number in the cells. When a new item $e$ arrives, the matched cell must have both a matched fingerprint and a matched sequence number. \reviewA{ With memory set as a constant, when the length of fingerprint increases, the risk of hash collisions of fingerprint decreases. But the number of buckets in the data structure decreases because each fingerprint takes more space. Therefore, the choice of the optimal length of fingerprints is made by a good trade-off between the above two effects. Besides, a longer fingerprint requires longer computation time. 
\presec
\section{Mathematical Analysis}
\postsec
\label{sec:math}

\begin{sloppypar}

\update{
In this section, we provide the theoretical analysis of the performance of \algo. 
% Due to space constraints, this section is provided in Appendix~\ref{sec:app:math}.
% In this section, we analyze the performance of \algo on a data stream.
% , and give some conclusion about the recall rate ~\ref{exp:setup:metrics} of our solutions.
}

% \update{
% In this section, we analyze the behavior of \textit{our hot panning version} on a single data stream, and prove that it meets \fairness{}.
% %
% We then give some conclusions about the error of the algorithm.
% %
% We also discuss how to apply the train of thought to the \textit{early freezing version}.
% }

% \newcommand{\preproof}{}
% \newcommand{\postproof}{}

% \presub
% \subsection{The estimation of recall}
% \postsub

\remove{
Since we only focus on a single data stream in this section, we decided to use some symbols that are simpler than the definitions in Section \ref{related::pre} to describe theorems and proofs more intuitively.
We first define the concepts of data stream and frequency.
Let $\mathcal{S}=\{e_1,e_2,\cdots,e_m\}$ be a data stream contains $m$ items, 
where $e_j \in \{1,2,\cdots,n\}$, and we say that item $e_j$ appears at time $j$.
Let $f_{(i,j)}=\sum_{k=1}^j 1_{\{e_k=i\}}$ be the frequency of item $i$ at time $j$, and let $f_i=f_{(i,m)}$ be the frequency of the item $i$ in data stream $\mathcal{S}$.
}

For convenience, we first define some variables.

\begin{itemize}
% \begin{itemize}[leftmargin=*]
    \item \textbf{Correct Instances (CI):} Number of correct instances.
    \item \textbf{Not-Reported Instances (NRI):} Number of unreported instances, $RR=1-\frac{NRI}{CI}$.
    \item \textbf{Recall Rate (RR):} Ratio of the number of correctly reported instances to the number of the correct instances.
\end{itemize}

Next, we assume:

\begin{itemize}

\item Flows are uniformly distributed in the $d$ buckets, which means each bucket might receive $t=\frac{N}{d}$ different flows, $N=\#(flow)$ is the number of flows. The flows also share the same $flow\ gap$ ratio $\beta$.

\item In each $A_{i}=\{flow\ |\ h(FID)\%d=i\}$, the next arriving item is independent to the existing items. Suppose the next item is $\langle FID,SEQ\rangle$ and a \textit{major gap} occurs, the \textit{major gap's} being reported indicates that the flow is contained in the existing items in the $i^{th}$ bucket.

%在每个A_i中，下一个到达的item和第i个bucket中的现有的item独立。假设下一个item是<fid,seq>,并且断流occur,则断流被报告说明第i个bucket现有的item中包含该flow

% \reviewA{
\item \reviewA{Flows' size follows $Zipf$ distribution. $Zipf$ distribution is a commonly used approximation of reality. Flow sizes may vary greatly. The items from the large flows (or the elephant flows defined in \cite{Newdirections}) is the vast majority and 
% should be focus on
deserves most attention\cite{Newdirections,elastic,k-elephant}.
% To simplify the distribution, 
% For simplicity, 
Therefore, we assume the size of the $j^{th}$ largest flow is $L_{0}*j^{-\alpha}(1 < \alpha \leq 3)$. 
% (This work mainly focuses on the condition of  $1 \leq \alpha \leq 3$).
}
\end{itemize}

%思路
Normally, $N \gg d,\ t \gg 1$. With $d$ over 1000 and $M\ll d$, the first $M$ large flows are typically evenly distributed in $M$ different sets $A_{j_1}, A_{j_2}, A_{j_3},\cdots, A_{j_M}$
% $\left\{A_j\right\}$. $A_{j_M}$
\reviewA{This is based on the independence of their FIDs' hash values.
It is just like $M$ balls randomly thrown into $d$ buckets. When $M\ll d$, the probability of $M$ balls thrown in different $M$ buckets is close to 1. In every $A_{j_i}$, the large flow is the vast majority since parameter $\alpha>1$.} Thus, we analyze one large flow and multiple small flows in each $A_{j_i}$.
% \reviewC{
\begin{lemma}
\label{lemma1}
    \reviewC{We denote the probability of i flows distributed in i different sets $A_{j_1}, A_{j_2}, A_{j_3},\cdots, A_{j_i}$ as $P_{diff}(i)$. For $M<d$ we have $\quad P_{diff}(M)>(1-\frac{M}{d})^{M-d}e^{-M}$.}
    % \reviewC{We denote the probability of i flows distributed in i different sets $A_j$ as $P_{diff}(i)$. Then the probability of $M(<d)$ large flows being distributed in $M$ different sets $A_j$ is $P_{diff}(M)$. We have $\quad P_{diff}(M)>(1-\frac{M}{d})^{M-d}e^{-M}$.}
\end{lemma}
% }
\begin{proof}
\reviewC{
For a good hash function, we can assume the independence and uniformity among the hashed values of different flow IDs.
}
\reviewC{
\begin{align*}
    \qquad \qquad \frac{P_{diff}(i+1)}{P_{diff}(i)}=P({i+1  \: \: differ  \: \: from 1,2,\dots,i})\\
\end{align*}
}
\vspace{-0.3in}
% \reviewC{
Therefore,
% }
\reviewC{
\begin{align*}
    \qquad \qquad \quad&P_{diff}(i+1)\\
    =&\frac{P_{diff}(i+1)}{P_{diff}(i)}\frac{P_{diff}(i)}{P_{diff}(i-1)}\cdots \frac{P_{diff}(2)}{P_{diff}(1)}P_{diff}(1) \\
    &(P_{diff}(1)=1)\\
        =&\Pi_{j=1}^{i} P({j+1  \: \: differ  \: \: from 1,2,\dots,j})\\
    =&\Pi_{j=1}^{i}\frac{d-j}{d}\\
    =&\Pi_{j=1}^{i} (1-\frac{j}{d})
\end{align*}
}
\vspace{-0.1in}
let $i+1=M$, we get
\reviewC{
\begin{align*}   
    \qquad \qquad&P_{diff}(M)\\
    =&\Pi_{j=1}^{M-1} (1-\frac{j}{d}) \\
    =&exp(d\sum_{j=1}^{M-1}\frac{1}{d}log(1-\frac{j}{d}))\\
    >&exp(d\sum_{j=1}^{M-1}\int_{j/d}^{j+1/d}log(1-x))\\
    =&exp(d\int_{1/d}^{M/d}log(1-x))\\
    >&exp(d\int_0^{M/d}log(1-x)dx)\\
    =&exp((M-d)log(1-M/d)-M)\\
    =&(1-\frac{M}{d})^{M-d}e^{-M}
\end{align*}
}
\reviewC{
Furthermore, for fixed $M$, we have
\begin{align*}
    \quad \lim \limits_{d \to \infty} (1-\frac{M}{d})^{M-d}e^{-M}=\lim \limits_{d \to \infty}(1-\frac{M}{d})^{-d}e^{-M}=e^{M}e^{-M}=1
\end{align*}
Therefore, we have $\lim \limits_{d \to \infty}P_{diff}(M)=1$.
}
\end{proof}

% \reviewA{
% The $M$ largest flows in the data stream are evenly distributed in different buckets is based on the independence of their FIDs' hash value.
% It is just like $M$ balls randomly thrown into $d$ buckets. When $d>>M$ the probability of $M$ balls thrown in different $M$ buckets is close to 1.
% }

Given a data stream and a sketch, the $j^{th}$ largest flow in $A_i$ is hereby denoted as $f_{i,j}$, while the size of $f_{i,j}$ is denoted as $L_{i,j}$.
\reviewC{
\begin{lemma}
\label{lemma2}
    The number of incorrectly reported instances of flow $f_{i,j}$ is smaller than $\beta*L_{i,j}(1-\frac{L_{i,j}}{\sum_{j} L_{i,j}})^w$, where $w$ is the number of cells in a bucket.
\end{lemma}
}

\begin{proof}
    If $f_{i,j}$ have already arrived before the arrival of $w$ other arrived items, the $major\ gap$ would be reported. In view of the weak correlation between the FIDs of two adjacent items, we can consider them to be independent. The probability of $f_{i,j}$ has not arrived in the last $w$ items is $(1-L_{i,j}/(\sum_{j}L_{i,j}))^w$ and then we have
% \begin{align*}
    $P(f_{i,j} \ arrived \ in \ w \ items)=1-(1-\frac{L_{i,j}}{\sum_{j} L_{i,j}})^w$.
% \end{align*}
    % $P(f_{i,j} \: \: arrived  \: \: in \: \: w \: \: items)=1-(1-\frac{L_{i,j}}{\sum_{j} L_{i,j}})^w$
    % $$P(f_{i,j} \: \: arrived  \: \: in \: \: w \: \: items)=1-(1-\frac{L_{i,j}}{\sum_{j} L_{i,j}})^w$$
    Therefore the NRI of flow $f_{i,j}$ is smaller than $\beta*L_{i,j}(1-\frac{L_{i,j}}{\sum_{j} L_{i,j}})^w$.
\end{proof}
% \reviewC{

% \vspace{-0.2in}
\begin{theorem}
\label{theo1}
For a data stream obeying Zipf distribution with $\alpha>1$ , we have $RR > 1-2^{\alpha}(\alpha-1)^{\frac{1}{\alpha}-1}M^{-\frac{(\alpha-1)^2}{\alpha}}$ with the probability of $(1-\tfrac{M}{d})^{M-d}e^{-M}$. \: \: 
\end{theorem}
% }
\begin{proof}
\reviewC{
    From Lemma~\ref{lemma1}, we know that with the probability of $(1-\frac{M}{d})^{M-d}e^{-M}$, the largest $M$ flows will be distributed in $M$ different sets $A_{j_1}, A_{j_2}, A_{j_3},\cdots, A_{j_M}$.
    Let $L_M=\frac{M^{1-\alpha}}{\alpha-1}$. We have $L_M=\int_{M}^{+\infty}x^{-\alpha}dx>\sum_{i=M+1}^{+\infty}i^{-\alpha}$.$\quad$Besides, we introduce a middle variable $M_1$. $M_1$ is defined as the largest integer such that $L_M<M_1^{-\alpha}$. Hence, we have $(M_1+1)^{-\alpha}<L_M<M_1^{\alpha}$.\\
}

    Assume the $i^{th}(1\leq i\leq M)$ largest flow is located in set $A_k$, $L_{k,1}=L_0*i^{-\alpha}$.\\
    
    Let $p_i=\frac{L_{k,1}}{\sum_{A_k}L_{k,j}}\quad$, consider$\quad \sum_{A_k}L_{k,j}=L_{k,1}+\sum_{j\neq1}L_{k,j}<L_{k,1}+\sum_{i=M}^{+\infty}L_0 *i^{-\alpha}<L_{k,1}+L_M*L_0$,we have$\quad p_i>\frac{L_{k,1}}{L_{k,1}+L_M*L_0}$.\\

    From ~\ref{lemma2}, we know that the NRI of $i^{th}$ flow is smaller than $\beta*L_{k,1}*(1-p_i)^w$, so:
\reviewC{
    \begin{align*}
        NRI&=\sum_{i=1}^{N} NRI_{i}\\
        &\leq \sum_{i=1}^{M_1}NRI_{i}+\sum_{i>M_1}\beta*L_{k,j}(related \: \: to \: \: i)\\
        &\leq        \beta(\sum_{i=1}^{M_1}L_0*i^{-\alpha}*(1-p_i)^w+\sum_{i>M_1} L_0*i^{-\alpha})\\
        &< \beta L_0(\sum_{i=1}^{M_1}i^{-\alpha}(\frac{L_M}{i^{-\alpha}+L_M})^w+\frac{M_1^{1-\alpha}}{\alpha-1})\\
        &< \beta L_0(\sum_{i=1}^{M_1} (L_M)^w * i^{\alpha(w-1)}+\frac{M_1^{1-\alpha}}{\alpha-1})\\
        &< \beta L_0((L_M)^{w}*\frac{M_1^{\alpha(w-1)+1}}{\alpha(w-1)+1}+\frac{M_1^{1-\alpha}}{\alpha-1})\\
        &<\beta L_0((M_1^{-\alpha})^{w}*\frac{M_1^{\alpha(w-1)+1}}{\alpha(w-1)+1}+\frac{M_1^{1-\alpha}}{\alpha-1})\\
        &<\beta L_0(\frac{M_1^{1-\alpha}}{\alpha(w-1)+1}+\frac{M_1^{1-\alpha}}{\alpha-1})\\
        &<\beta L_0(\frac{2M_1^{1-\alpha}}{\alpha-1})\\
        &= \beta L_0\frac{2(M_1+1)^{1-\alpha}(\frac{M_1}{M_1+1})^{1-\alpha}}{\alpha-1}\\
        &<\beta L_0 \frac{2^{\alpha}L_M^{-\frac{1-\alpha}{\alpha}}}{\alpha-1}=\frac{\beta L_0}{\alpha-1}\frac{2^{\alpha}M^{-\frac{(\alpha-1)^2}{\alpha}}}{(\alpha-1)^{1-\frac{1}{\alpha}}}
    \end{align*}
}
\reviewC{
    With $CI=\beta(\sum_{i}L_0*i^{-\alpha})>\beta(\int_{1}^{+\infty}\frac{x^{-\alpha}}{\alpha-1})= \tfrac{\beta L_0}{\alpha-1}$ , we have
    }
\reviewC{
    $$
    RR=1-\frac{NRI}{CI}>1-2^{\alpha}(\alpha-1)^{\frac{1}{\alpha}-1}M^{-\frac{(\alpha-1)^2}{\alpha}}
    $$
    }
\end{proof}

\reviewC{
% For $1<\alpha<3$, the constant $2^{\alpha}(\alpha-1)^{\frac{1}{\alpha}-1}$ is not too large (more precisely $2^{\alpha}(\alpha-1)^{\frac{1}{\alpha}-1}<6$). $\lim_{M \to \infty}RR=1$, and for $\alpha>1.5$, RR converges rapidly to 1 which provide explanation for some results in experiments.
For $1<\alpha<3$, the constant $2^{\alpha}(\alpha-1)^{\frac{1}{\alpha}-1}<6$. Therefore, $\lim \limits_{M \to \infty}RR=1$. 
% Besides, for $\alpha>1.5$, RR converges rapidly to 1, which provide explanation for some results in experiments.
}   
% \end{proof}

% The value of $M$ is selected by balancing probability and estimated value. With $d$ large enough, an adequate choice of $M$ will provide both a probability and an estimate value of RR close enough to 1.

% We then define the state $s_{(k,t)}$ of the Double-Anonymous sketch on data stream $\mathcal{S}_k$ at time $t$ as $s_{(k,t)} = \{s_{(k,1,t)}, \cdots, s_{(k,n_k,t)}\}$, where $ s_{(k,i,t)}= \langle f_{T(k,i,t)},f_{S(k,i,t)} \rangle $.
% %
% In general, let $f_{T(k,i,t)}$ be the frequency of item $u_{(k,i)}$ recorded in the top-$K$ part at time $t$, and let $f_{S(k,i,t)}$ be the frequency of item $i$ recorded in the count part at time $t$.
% %
% In particular, if item $u_{(k,i)}$ is not recorded in the top-$K$ part at time $j$, let $f_{T(k,i,t)}=0$.

% Given a data stream $\mathcal{S}_k$, let a \textit{sketching process} $\mathcal{R}$ be a sequence of states of the Double-Anonymous sketch at each time, \ie, $\mathcal{R}=\{s_{(k,1)},s_{(k,2)},\cdots,s_{(k,m_k)}\}$.
% %
% The replacement policy $\mathcal{P}$ determines the distribution of the sketching process, \ie, $\mathcal{R}\sim\mathcal{P}(\mathcal{S}_k)$.

\end{sloppypar}

\presec
\section{Experiments}
\postsec
\label{sec:exp}

\begin{figure*}[!ht]
	\centering
        \subfigure[CAIDA]{
		\begin{minipage}[t]{0.23\textwidth}{
		\begin{center}
		\includegraphics[width=\textwidth]{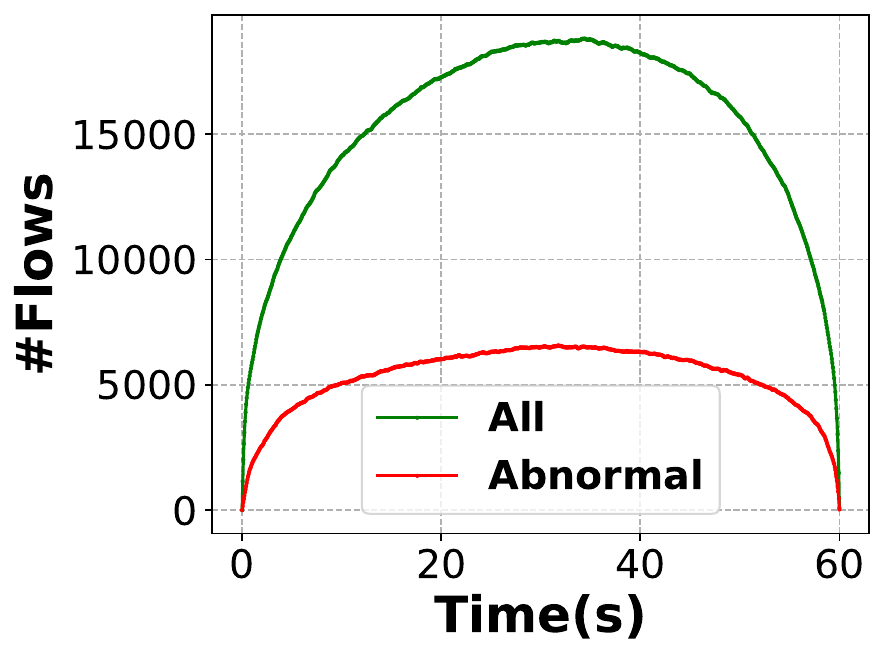}
		\end{center}
		}
		\label{concur-CAIDA}
		\end{minipage}
	}
	\subfigure[MAWI]{
		\begin{minipage}[t]{0.23\textwidth}{
		\begin{center}
		\includegraphics[width=\textwidth]{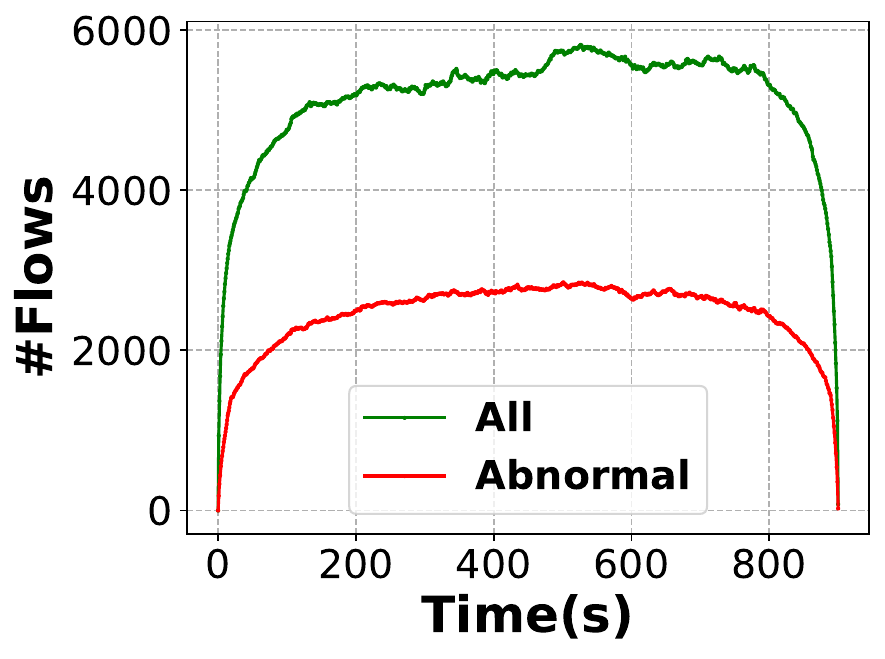}
		\end{center}
		}
		\label{concur-MAWI}
		\end{minipage}
	}
	\subfigure[MACCDC]{
		\begin{minipage}[t]{0.23\textwidth}{
		\begin{center}		
		\includegraphics[width=\textwidth]{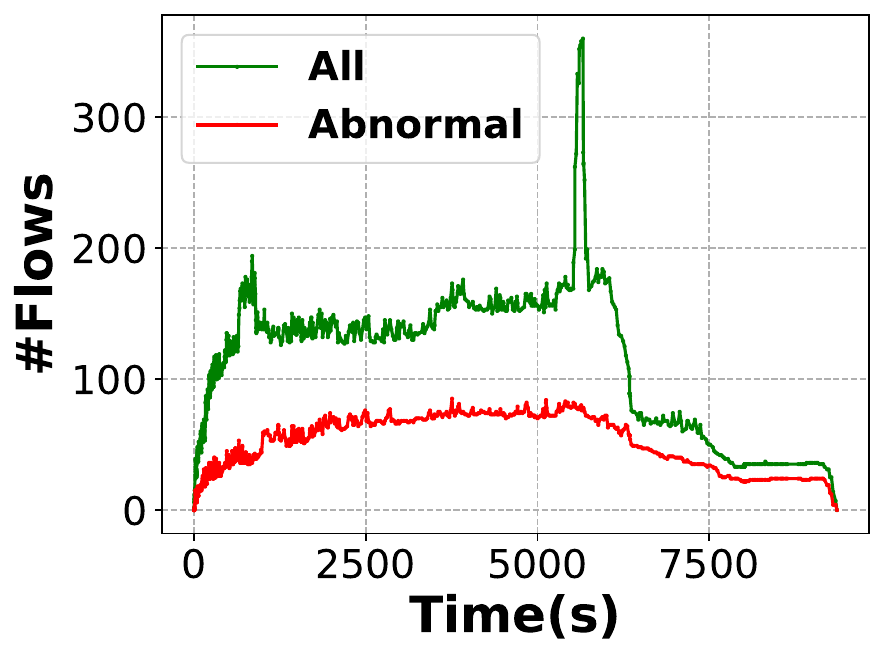}
		\end{center}
		}
		\label{concur-MACCDC}
		\end{minipage}
	}
	\subfigure[IMC]{
		\begin{minipage}[t]{0.23\textwidth}{
		\begin{center}		
		\includegraphics[width=\textwidth]{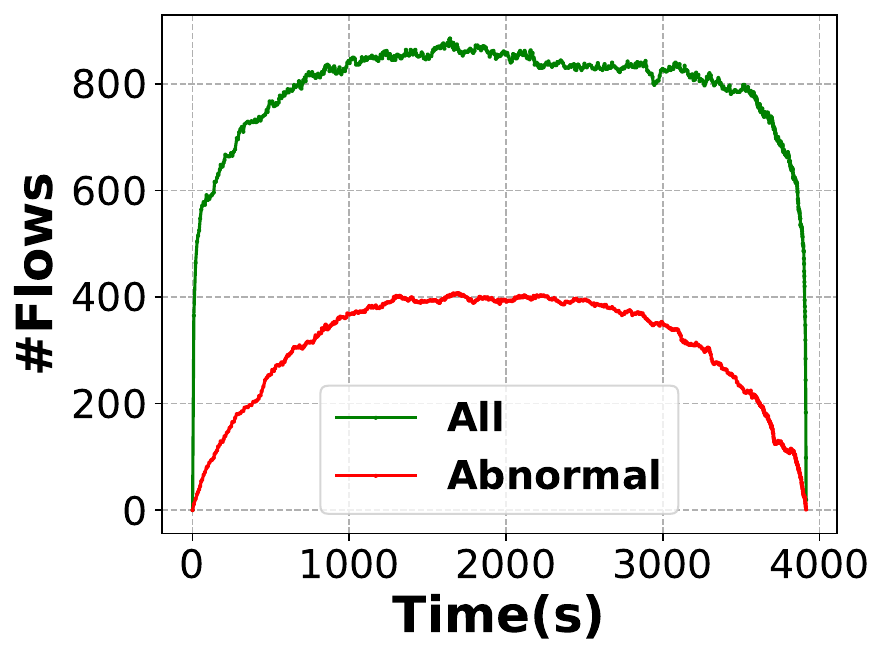}
		\end{center}
		}
		\label{concur-IMC}
		\end{minipage}
	}
 %\vspace{-0.1in}
    \caption{Concurrency circumstances of all/abnormal flows on different datasets.}
    \label{concur}
    %\vspace{-0.08in}
\end{figure*}
\begin{figure}[!ht]
	\centering
        \subfigure[Flow Length]{
		\begin{minipage}[t]{0.225\textwidth}{
		\begin{center}
		\includegraphics[width=\textwidth]{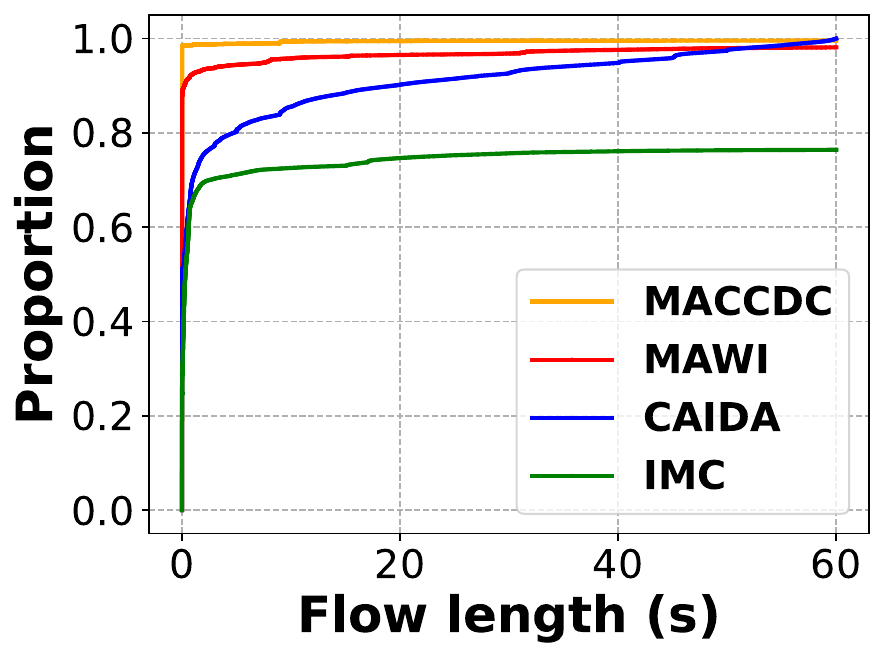}
		\end{center}
		}
		\label{fig:datasetDistri:length}
		\end{minipage}
	}
	\subfigure[\textit{Gap} Size]{
		\begin{minipage}[t]{0.225\textwidth}{
		\begin{center}
		\includegraphics[width=\textwidth]{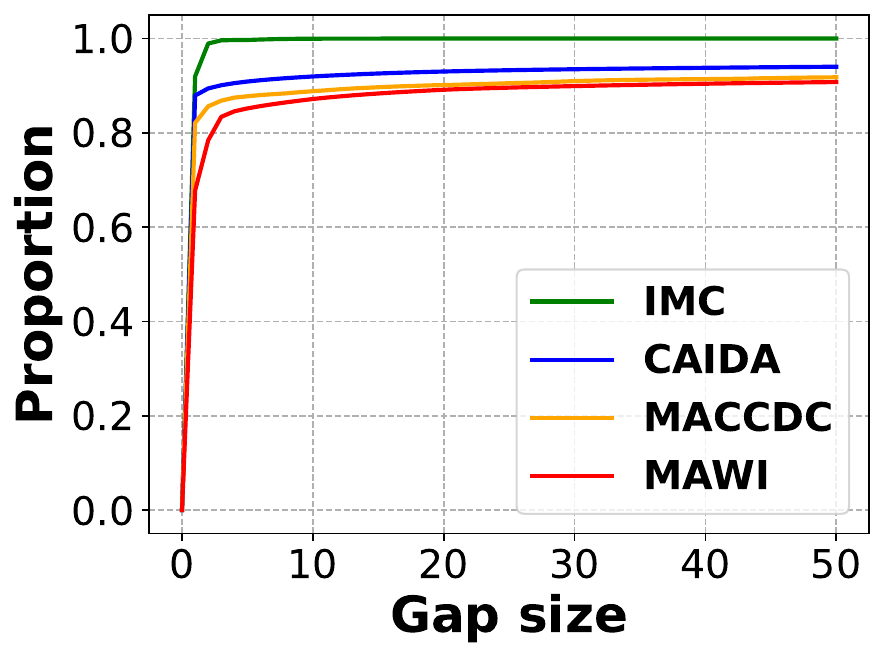}
		\end{center}
		}
		\label{fig:datasetDistri:gap}
		\end{minipage}
	}
 	
    \caption{Distributions of flow length and \textit{gap} size.}
    \label{fig:datasetDistri}
    %\vspace{-0.08in}
\end{figure}

\begin{sloppypar}

In this section, we evaluate the performance of \algo{} through experiments.
We implement \algo{} on the CPU platform, and evaluate it in the $flow\ gap$ problem.
The source codes are available at GitHub \cite{link}. 

% Following key results are observed: (1)  . (2)  .
\vspace{-0.15in}
\subsection{Experiment Setup}
\subsubsection{Datasets.} \label{exp:dataset} ~

We select four real-world datasets and perform experiments using the Identification field in the IP header as the sequence number. \reviewA{The concurrency circumstances of all/abnormal flows in these datasets are illustrated in Figure~\ref{concur}. The distributions of flow length and \textit{gap} size in these datasets are illustrated in Figure~\ref{fig:datasetDistri}. We set $\mathcal{T}_1=5,\mathcal{T}_2=30$.}

\begin{itemize}[leftmargin=*]
    \item \textbf{CAIDA:} We use an anonymized network trace dataset collected by CAIDA \cite{caida} in 2018. Each item in the dataset is distinguished by a 5-tuple (source IP address, source port, destination IP address, destination port, protocol) that uniquely identifies a UDP/TCP session. The slice used in this work contains network traffic in 1 min, which includes around 30M items and 1.3M flows.
    \item \textbf{MAWI:} The MAWI dataset contains real traffic trace data maintained by the MAWI Working Group\cite{mawi}. Each item ID in this dataset is also a 5-tuple, similar to the CAIDA dataset. There are around 55M items and 9M flows in the MAWI dataset.
    \item \textbf{MACCDC:} The MACCDC dataset, comprising of approximately 3M items and 2M flows, is provided by the U.S. National CyberWatch Mid-Atlantic Collegiate Cyber Defense Competition (MACCDC)\cite{maccdc}.
    \item \textbf{IMC:} The IMC dataset is sourced from one of the data centers studied in \emph{ Network Traffic Characteristics of Data Centers in the Wild}\cite{benson2010network}. Each item is identified by a 5-tuple. There are around 18M items in the IMC dataset, with 560K total flows.
\end{itemize}

\bbb{Synthetic item loss.}
To validate the effectiveness of \algo under extreme network conditions, we design a method of imposing synthetic item loss and apply it to our original datasets. In the experiments, the CAIDA, MAWI, and MACCDC datasets are used without synthetic item loss, whereas the IMC dataset is imposed with synthetic item loss. 
% The process of imposing synthetic item loss and the experiments testing the robustness of \algo by applying synthetic item loss are shown in Appendix~\ref{sec:exp:synthetic}.

We divide our dataset into $n_T$ time windows of equal lengths and categorize flows within each window as normal or abnormal. For each time window, a certain portion $r\in(0,1)$ of the flows are randomly chosen and marked as abnormal, with the remaining flows marked as normal. We define two kinds of item loss in the data stream: (1) Consecutive item loss and (2) Single item loss. Consecutive item loss is implemented only on the abnormal flows. For any item $e_1= \langle FID_1,SEQ_1 \rangle $ in the abnormal flows, we generate a random number $j \in \{x \in \mathbb{N} | \mathcal{T}_1 \leq x < \mathcal{T}_2\}$ with equal probability. In the flow $FID_1$, we drop all the items with $SEQ \in \{x \in \mathbb{N} | {SEQ_1} \leq x < {SEQ_1}+j\}$ with a probability of $b^j$, where $b$ is a predefined constant. For a normal flow or an abnormal flow escaping the consecutive item loss, we execute the single item loss on it. For any item $e_2= \langle FID_2,SEQ_2 \rangle $ in the flows applicable to single item loss, we drop $e_{2}$ with a probability of a predefined number $p$. 
% In sum, the normal flows only experience single item loss. The abnormal flows first experience consecutive item loss. They only experience single item loss when they escape consecutive item loss.

Consecutive item loss simulates the network congestion observed in real-world scenarios, where the buffer in a router or switch fills up and subsequent items arriving at this node are dropped. If such situation persists, consecutive item loss would occur in a flow. Single item loss represents items being lost during the transmission process due to weak or unstable signal conditions in the real world.
% poor signal quality in real world. Items may be lost in transmission due to weak or unstable signal.

\vspace{0.3in}
\subsubsection{Implementation.}~
\label{exp:imple}

We implement our \algo{}-SO, \algo{}-AO and Straw-man solution in C++ on a CPU platform. The hash functions utilized are the 64-bit Bob Hash \cite{bobhash} initialized with different random seeds. \reviewB{Both \algo{}-SO and \algo{}-AO utilize sequence number randomizing. Only \algo{}-AO utilizes fingerprint.}
% We set $r=0.4,b=0.5,\mathcal{T}_1=5,\mathcal{T}_2=30$, which are parameters used in industry.

% For real-time report scenario, we report a $major \ gap$ immediately upon detection. The reported content is a key-value pair  $\langle FID,SEQ \rangle$, 
% signifying which flow is abnormal and where the $major \ gap$ happens. We use Precision Rate, Recall Rate and $F_1$ Score (see Section~\ref{exp:setup:metrics} for definition) to measure the accuracy.

% For periodic report scenario, we report the abnormal flows and the occurrence number of $major$ $gap$s for each abnormal flow at the end of every time window. We use Mean Absolute Error and Root Mean Square Error (see Section~\ref{exp:setup:metrics} for definition) to measure the accuracy. We use a linear probing hash table to temporarily store the occurrence number of $major \ gap$s for each abnormal flow in every time window.
% % , before the end of the time window is reached. 
% The hash table is cleared at the end of every time window to prepare for the next window.
% In the following experiments, the memory cost of this hash table is omitted, since it is not where our emphasis is put.

\vspace{0.7in}
\subsubsection{Metrics.} \label{exp:setup:metrics}

\begin{itemize}[leftmargin=*]
    \item \textbf{Precision Rate (PR):} Precision rate is the ratio of the number of correctly reported instances to the number of reported instances.
    \item \textbf{Recall Rate (RR):} Recall rate is the ratio of the number of correctly reported instances to the number of the correct instances.
    \item \textbf{$F_1$ Score:} $\frac{2 \times PR \times RR}{PR+RR}$
    % \item \textbf{Mean Absolute Error (MAE):} At the end of every time window, we calculate the mean absolute error of the frequency of $major \ gap$ of every flow. 
    % We define the mean absolute error as MAE = $\frac{1}{|\Psi|}\sum\limits_{f_i \in \Psi}|Occ_i-\hat{Occ_i}|$, where $Occ_i$ is the real occurrence number of $major \ gap$ of flow $f_i$, $\hat{Occ_i}$ is the estimated occurrence number of $major \ gap$ of flow $f_i$, and $\Psi$ is the union set of real abnormal flows and reported abnormal flows in the time window.
    % \item \textbf{Root Mean Square Error (RMSE):} We define the root mean square error as RMSE = $\sqrt[]{\frac{1}{|\Psi|}\sum\limits_{f_i \in \Psi}(Occ_i-\hat{Occ_i})^2}$. $Occ_i$ and $\hat{Occ_i}$ are the real occurrence number and the estimated occurrence number of $major$ $gap$ of flow $f_i$, respectively. $\Psi$ is the union set of real abnormal flows and reported abnormal flows in the time window.
\end{itemize}

\vspace{0.3in}
\subsubsection{Algorithm Comparison}~
\label{strawman}

% \reviewA{Cuckoo Filter}

We devised a Straw-man solution based on Cuckoo Filter \cite{fan2014cuckoo}. Cuckoo Filter is an efficient hash table implementation based on cuckoo hashing \cite{pagh2004cuckoo}, which can achieve both high utilization and compactness. It realizes constant time lookup and amortization constant time insertion operations. Straw-man also uses fingerprint instead of recording flow-id to improve space utilization.

\reviewA{
Specifically, the data structure consists of a table of buckets and three hash functions $h_1(\cdot), h_2(\cdot)$, and $h_f(\cdot)$. Every bucket contains $w$ cells, each recording a fingerprint $fp^c$ and a sequence number $SEQ^c$. Due to the fact that $w$ in the original design of Cuckoo Filter must be a power of two, we improve it by dividing buckets into two blocks to accommodate various memory sizes: For each incoming item $e = \langle FID, SEQ \rangle $, we first calculate its $fp= h_f(FID)$ and map it into the $[h_1(fp)]^{th}$ bucket in block one and the $[h_2(fp)]^{th}$ bucket in block two. Then we search for a cell with its $fp^c=fp$ in the two buckets. With the matched cell found, we calculate $dif=SEQ-SEQ^c$, use equation (\ref{situation}) to determine the situation, and update the $SEQ^c$ to be $max\{SEQ,SEQ^c\}$. If such cells do not exist, we will insert $e$ to a cell by setting its $FID^c=FID$ and $fp^c=fp$ when a bucket still contains empty cells. If all the buckets are full, we will randomly evict an item in a cell, and place it into the other bucket where it can go. If that bucket is also full, then another eviction will be triggered. This process goes on and on until an empty cell exists or a predefined \textit{MAX\_NUMBER\_OF\_TURNS} is reached. In these experiments, we set $w=4$ and length of $fp$ to be $32bits$ to avoid collisions. \textit{MAX\_NUMBER\_OF\_TURNS} is set to $8$ because it can already achieve a high utility rate of memory. Setting a larger \textit{MAX\_NUMBER\_OF\_TURNS} will not improve the accuracy much but will lower the throughput.
% nearly the highest utilization.  
}

% {The most straightforward Straw-man solution is to use a hash table to record the flow ID and the sequence number of the incoming item. Specifically, the data structure consists of a separate chaining hash table of buckets and a hash function $h(\cdot)$. Every bucket contains a linked list of cells, each recording a flow ID $FID^c$ and a sequence number $SEQ^c$. Each incoming item $e =  \langle FID, SEQ \rangle $ is mapped into the $[h(FID)]^{th}$ bucket. If the linked list in it is empty, we will initialize a cell by setting its $FID^c=FID$ and $SEQ^c=SEQ$. Otherwise, we will search for a cell with its $FID^c=FID$ along the linked list. With the matched cell found, we calculate $dif=SEQ-SEQ^c$, use equation (\ref{situation}) to determine the situation, and update the $SEQ^c$ to be $max\{SEQ,SEQ^c\}$. If such cells do not exist, we will append a cell with its $FID^c=FID$ and $SEQ^c=SEQ$ into the list when memory allows. When the consumed memory reaches the limit, we will randomly empty a cell in the linked list and insert $e$ into it.}

% 

% \subsubsection{Default Settings}~

% We set the parameters to the following default values, unless the experiments specifically measure the impact of certain parameters on the algorithm's performance. We set $\mathcal{T}_1=5,\mathcal{T}_2=30$. 
% % The data type used to store SEQ is uint16\_t, which occupies 2 bytes (16 bits) of memory.  Has been told in Dataset part
% In \algo-SO, $w=4$. In \algo-AO, $w=8$, with $s=2$ and $c=6$. The length of fingerprint for \algo-AO $l_f=8 bits$.

\vspace{-0.1in}
\subsection{Experiments on Parameter Settings}
\label{art:exp:para}
% \vspace{-0.3in}

% In this section, we conduct experiments to evaluate how the number of cells ($w$) in one bucket affects the accuracy of our \algo{}-SO, how the length of fingerprint affects the accuracy of our \algo{}-AO, and how the memory ratio of $suspect$ to $civilian$ affects the accuracy of our \algo{}-AO. 

\begin{figure}[!ht]
	\centering
        \subfigure[CAIDA]{
		\begin{minipage}[t]{0.225\textwidth}{
		\begin{center}
		\includegraphics[width=\textwidth]{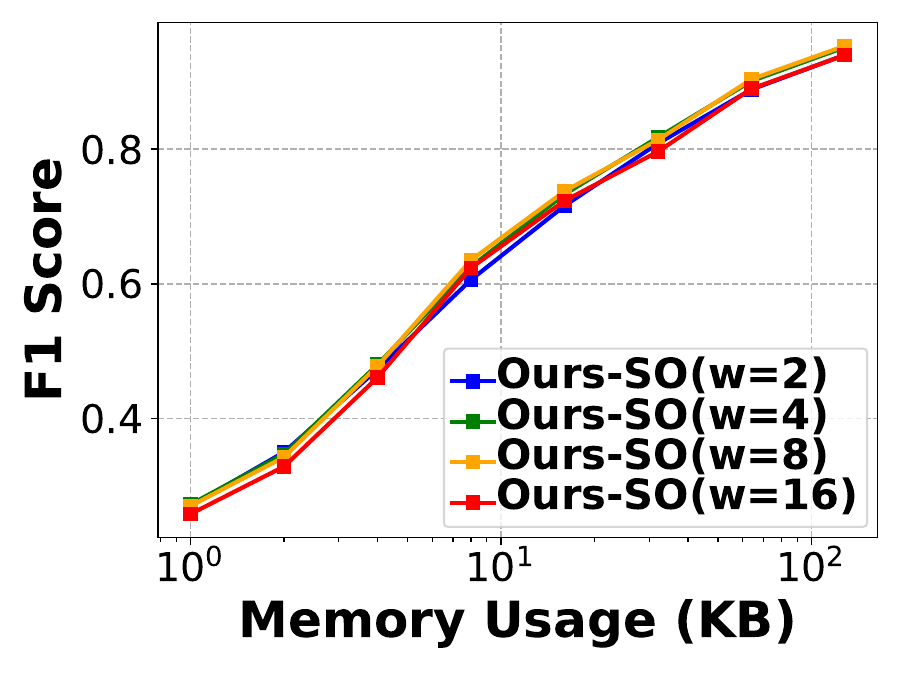}
		\end{center}
		}
		\label{para-SOwidth-CAIDA-F1}
		\end{minipage}
	}
	\subfigure[MAWI]{
		\begin{minipage}[t]{0.225\textwidth}{
		\begin{center}
		\includegraphics[width=\textwidth]{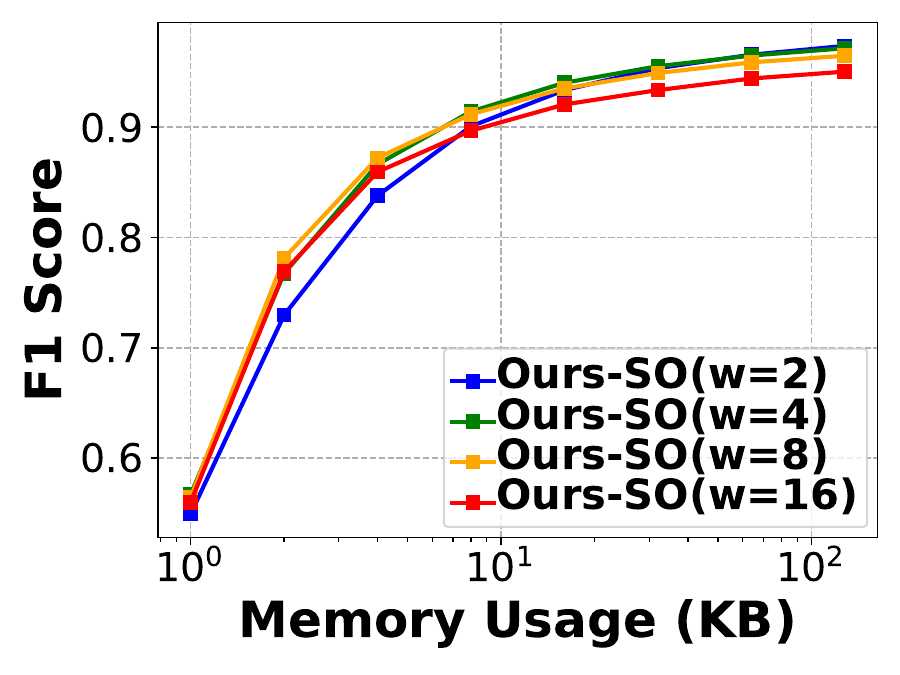}
		\end{center}
		}
		\label{para-SOwidth-MAWI-F1}
		\end{minipage}
	}
	%
	% \subfigure[IMC]{
	% 	\begin{minipage}[t]{0.225\textwidth}{
	% 	\begin{center}		
	% 	\includegraphics[width=\textwidth]{Figures/exp/tmp.pdf}
	% 	\end{center}
	% 	}
	% 	\label{para-SOwidth-IMC-RMSE}
	% 	\end{minipage}
	% }
 % \vspace{-0.15in}
    \caption{Effect of $w$ on \algo{}-SO on different datasets.}
    \label{para-SOwidth}
    % \vspace{-0.30in}
\end{figure}

\begin{figure}[!ht]
	\centering
        \subfigure[CAIDA]{
		\begin{minipage}[t]{0.225\textwidth}{
		\begin{center}
		\includegraphics[width=\textwidth]{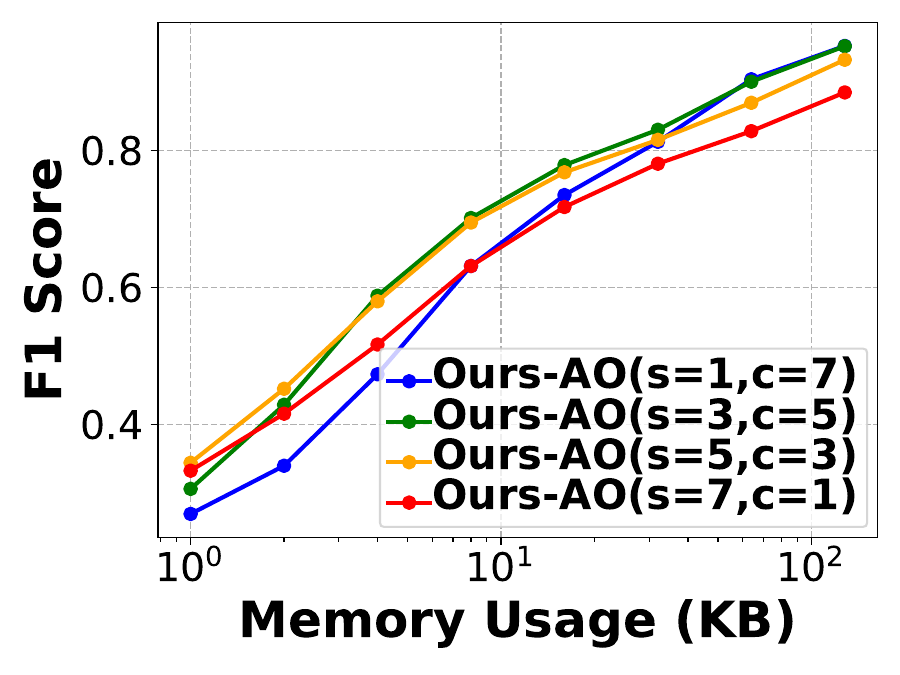}
		\end{center}
		}
		\label{para-AOratio-CAIDA_F1}
		\end{minipage}
	}
	\subfigure[MAWI]{
		\begin{minipage}[t]{0.225\textwidth}{
		\begin{center}
		\includegraphics[width=\textwidth]{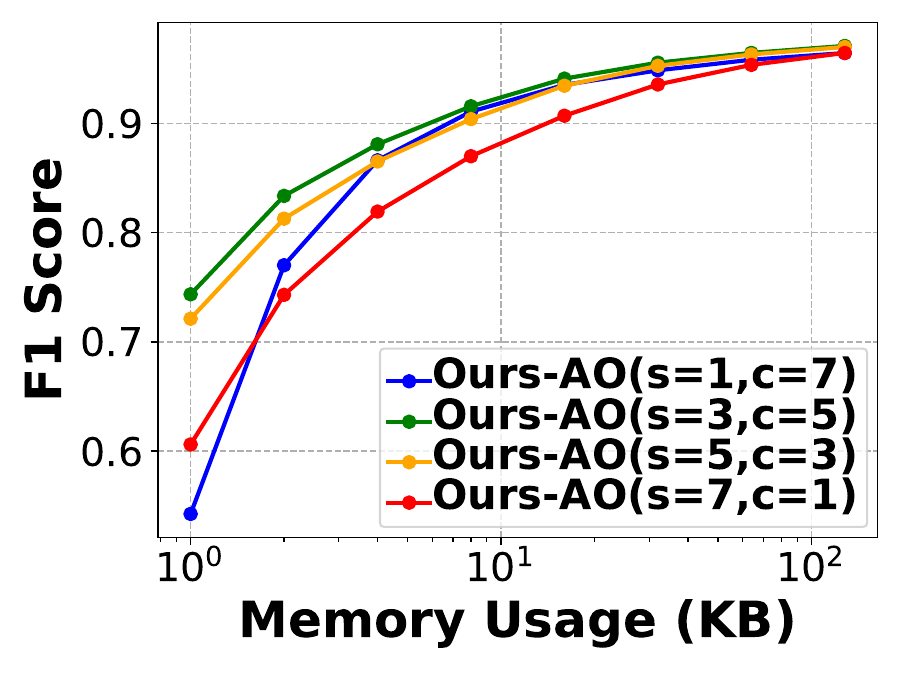}
		\end{center}
		}
		\label{para-AOratio-MAWI_F1}
		\end{minipage}
	}
	%
	% \subfigure[IMC]{
	% 	\begin{minipage}[t]{0.225\textwidth}{
	% 	\begin{center}		
	% 	\includegraphics[width=\textwidth]{Figures/exp/tmp.pdf}
	% 	\end{center}
	% 	}
	% 	\label{para-AOratio-IMC-RMSE}
	% 	\end{minipage}
	% }
 % \vspace{-0.15in}
    \caption{Effect of \textit{s/c} on \algo{}-AO on different datasets.}
    \label{para-AOratio}
    % \vspace{-0.30in}
\end{figure}
\begin{figure}[!ht]
	\centering
        \subfigure[MACCDC]{
		\begin{minipage}[t]{0.225\textwidth}{
		\begin{center}
		\includegraphics[width=\textwidth]{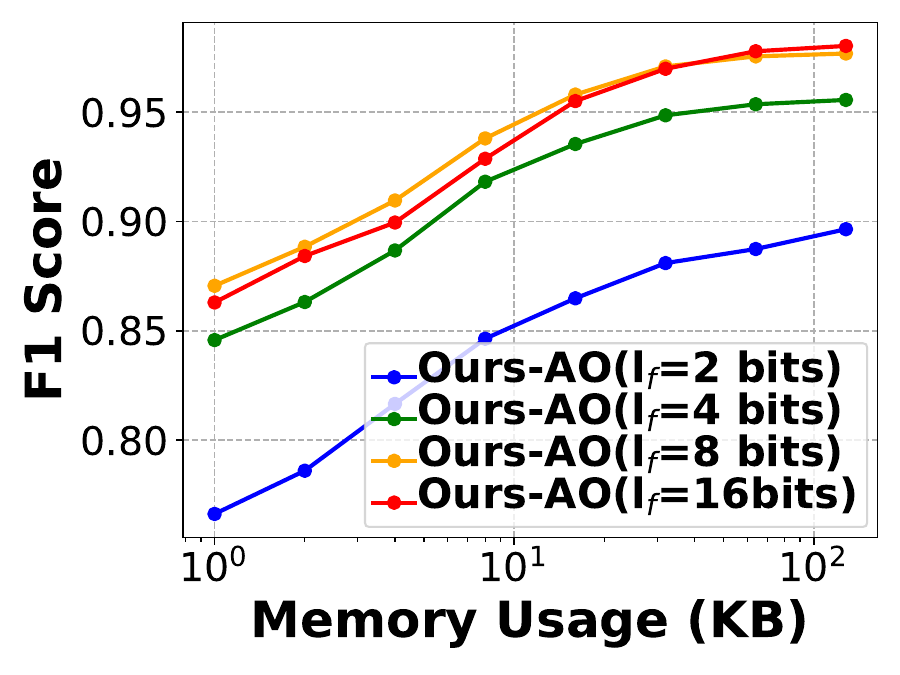}
		\end{center}
		}
		\label{para-AOfp-MACCDC-F1}
		\end{minipage}
	}
	\subfigure[IMC]{
		\begin{minipage}[t]{0.225\textwidth}{
		\begin{center}
		\includegraphics[width=\textwidth]{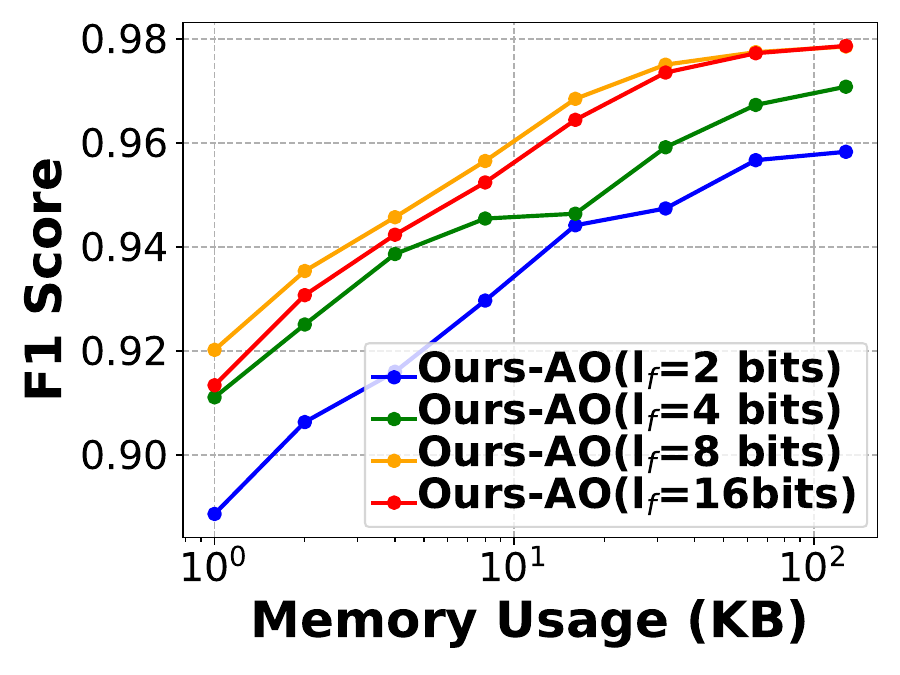}
		\end{center}
		}
		\label{para-AOfp-IMC-F1}
		\end{minipage}
	}
	%
	% \subfigure[IMC]{
	% 	\begin{minipage}[t]{0.225\textwidth}{
	% 	\begin{center}		
	% 	\includegraphics[width=\textwidth]{Figures/exp/tmp.pdf}
	% 	\end{center}
	% 	}
	% 	\label{para-AOfp-IMC-RMSE}
	% 	\end{minipage}
	% }
 % \vspace{-0.15in}
    \caption{Effect of $l_f$ on \algo{}-AO on different datasets.}
    \label{para-AOfp}
    % \vspace{-0.10in}
\end{figure}
\begin{figure*}[!ht]
	\centering
        \subfigure[CAIDA]{
		\begin{minipage}[t]{0.23\textwidth}{
		\begin{center}
		\includegraphics[width=\textwidth]{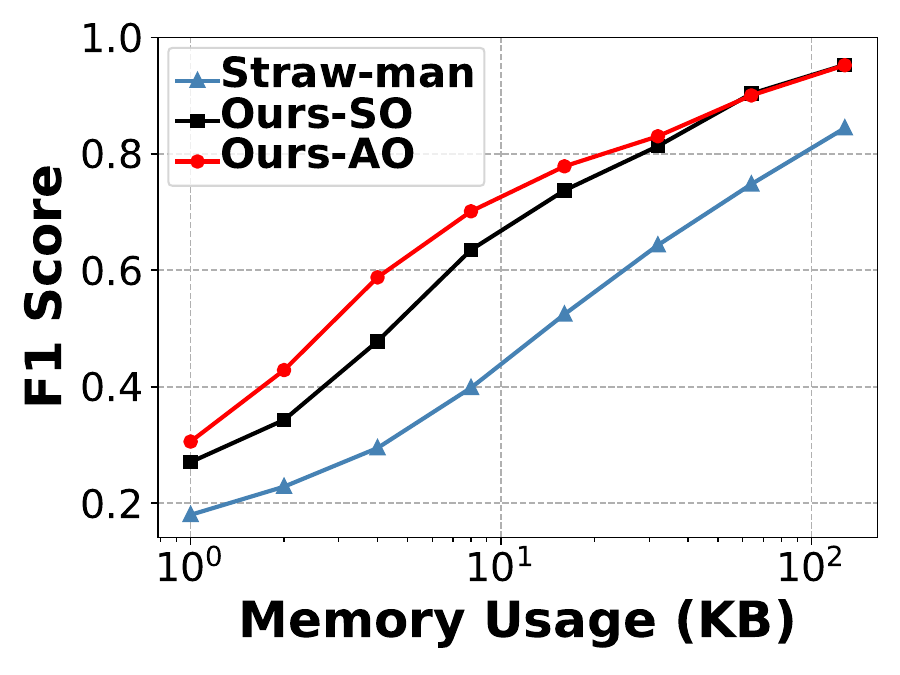}
		\end{center}
		}
		\label{real-F1-CAIDA}
		\end{minipage}
	}
	\subfigure[MAWI]{
		\begin{minipage}[t]{0.23\textwidth}{
		\begin{center}
		\includegraphics[width=\textwidth]{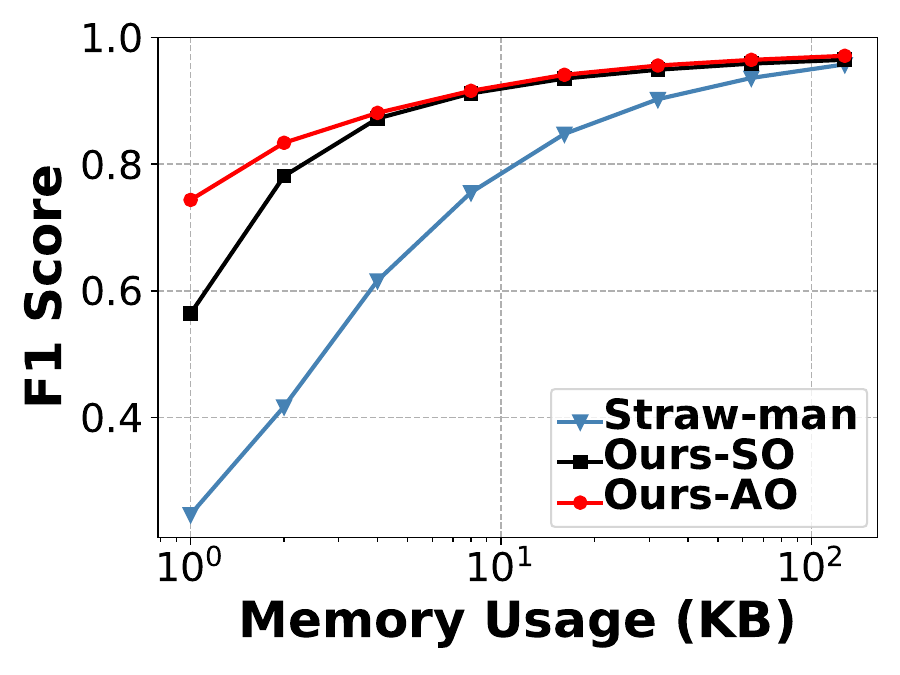}
		\end{center}
		}
		\label{real-F1-MAWI}
		\end{minipage}
	}
	\subfigure[MACCDC]{
		\begin{minipage}[t]{0.23\textwidth}{
		\begin{center}		
		\includegraphics[width=\textwidth]{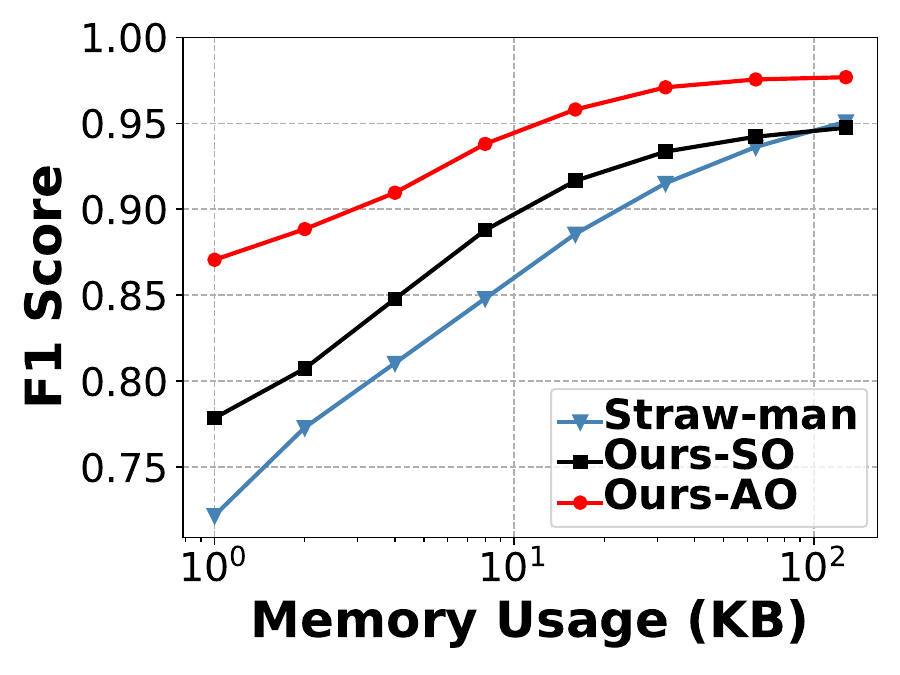}
		\end{center}
		}
		\label{real-F1-MACCDC}
		\end{minipage}
	}
	\subfigure[IMC]{
		\begin{minipage}[t]{0.23\textwidth}{
		\begin{center}		
		\includegraphics[width=\textwidth]{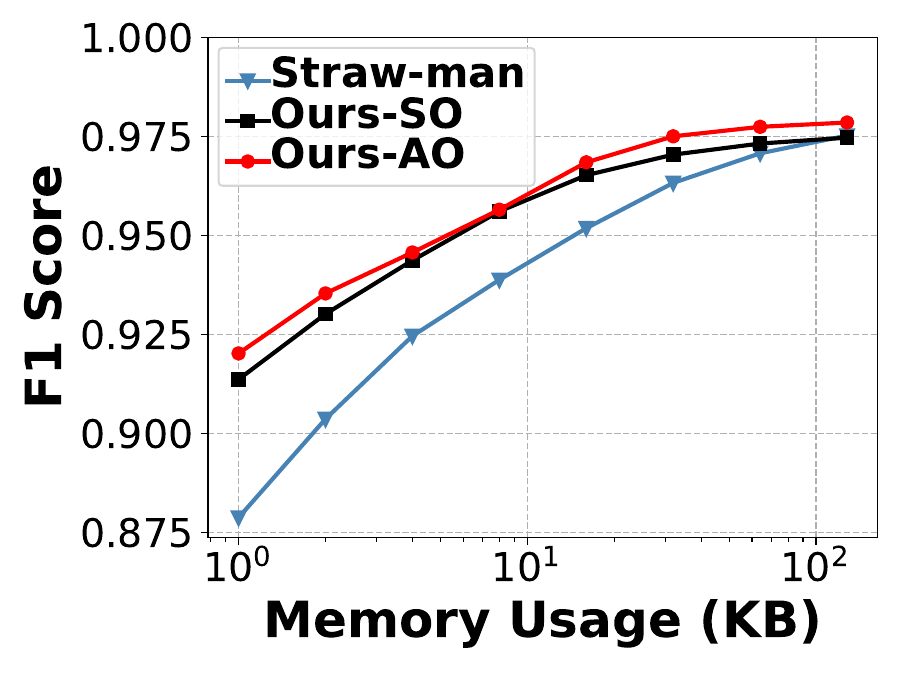}
		\end{center}
		}
		\label{real-F1-IMC}
		\end{minipage}
	}
 %\vspace{-0.1in}
    \caption{$F_1$-Score of algorithms on different datasets.}
    \label{real-F1}
    % \vspace{-0.30in}
\end{figure*}
\begin{figure*}[!ht]
	\centering
        \subfigure[CAIDA]{
		\begin{minipage}[t]{0.225\textwidth}{
		\begin{center}
		\includegraphics[width=\textwidth]{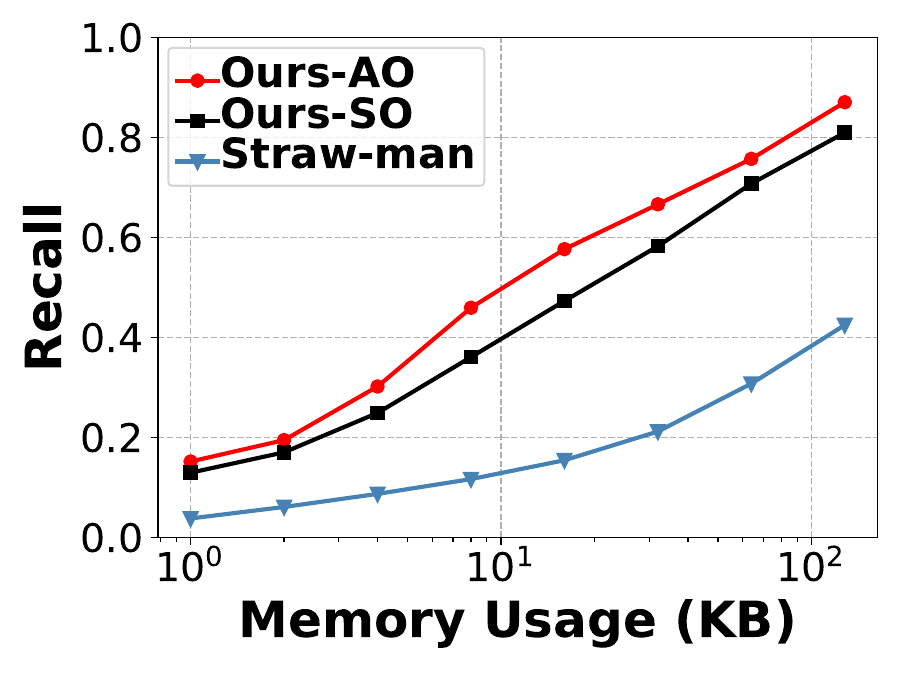}
		\end{center}
		}
		\label{real-RR-CAIDA}
		\end{minipage}
	}
	\subfigure[MAWI]{
		\begin{minipage}[t]{0.225\textwidth}{
		\begin{center}
		\includegraphics[width=\textwidth]{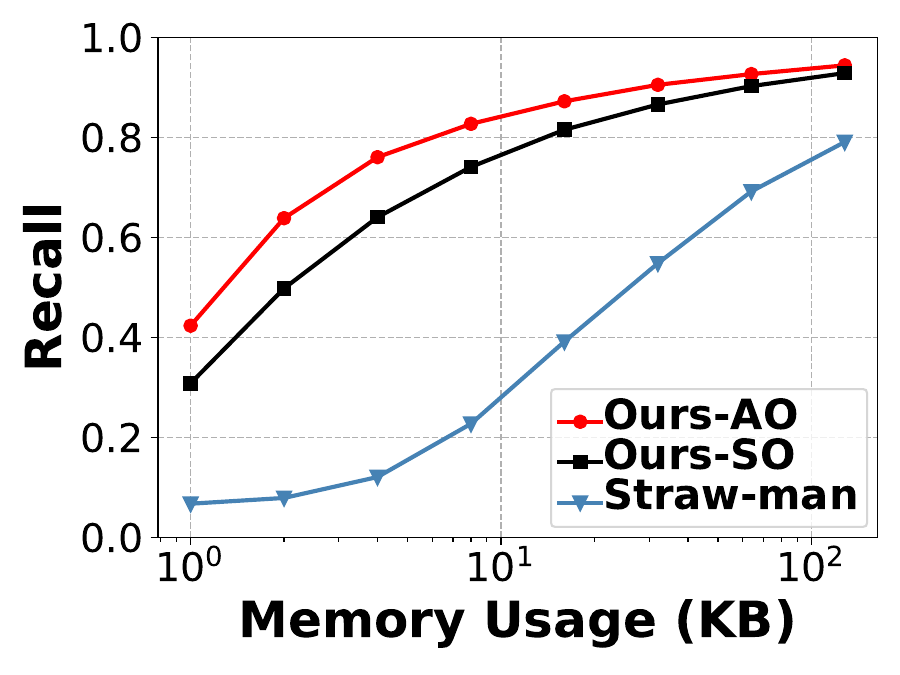}
		\end{center}
		}
		\label{real-RR-MAWI}
		\end{minipage}
	}
	\subfigure[MACCDC]{
		\begin{minipage}[t]{0.225\textwidth}{
		\begin{center}		
		\includegraphics[width=\textwidth]{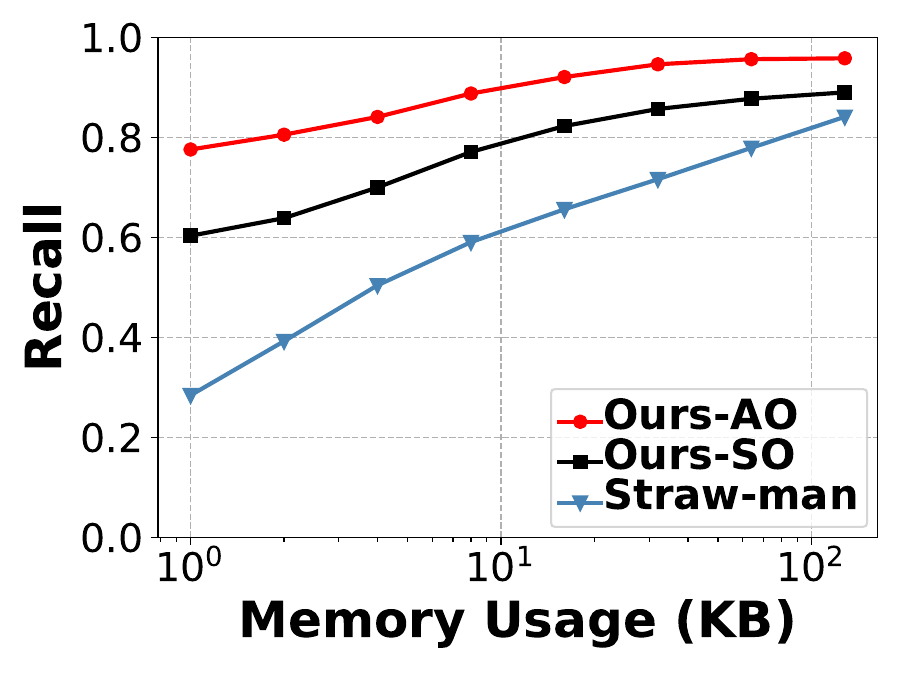}
		\end{center}
		}
		\label{real-RR-MACCDC}
		\end{minipage}
	}
	\subfigure[IMC]{
		\begin{minipage}[t]{0.225\textwidth}{
		\begin{center}		
		\includegraphics[width=\textwidth]{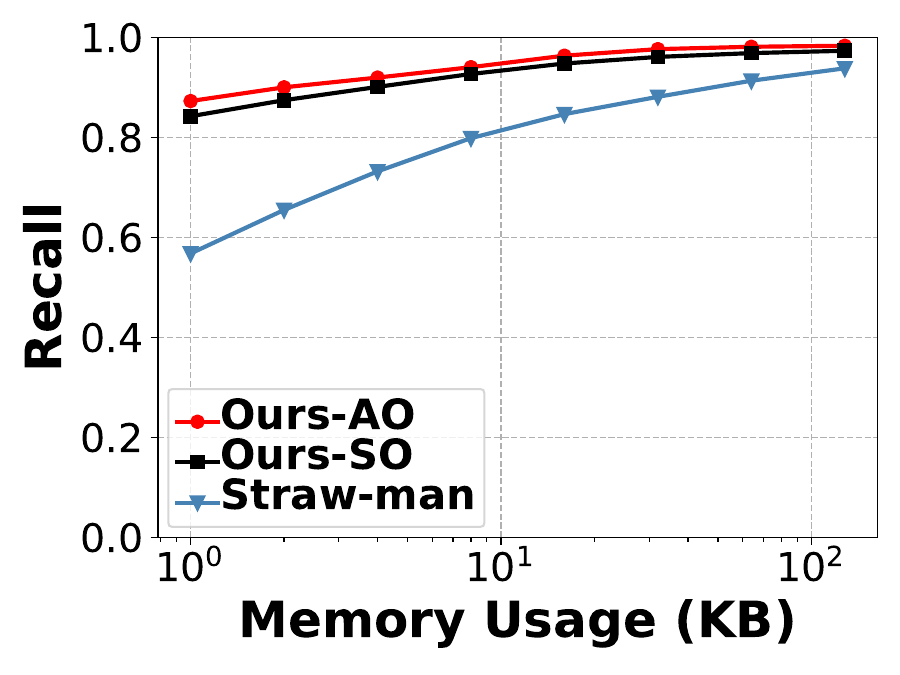}
		\end{center}
		}
		\label{real-RR-IMC}
		\end{minipage}
	}
 %\vspace{-0.1in}
    \caption{RR of algorithms in Real-time Report Scenario on different datasets.}
    \label{real-RR}
    %\vspace{-0.08in}
\end{figure*}
\begin{figure*}[!ht]
	\centering
        \subfigure[CAIDA]{
		\begin{minipage}[t]{0.225\textwidth}{
		\begin{center}
		\includegraphics[width=\textwidth]{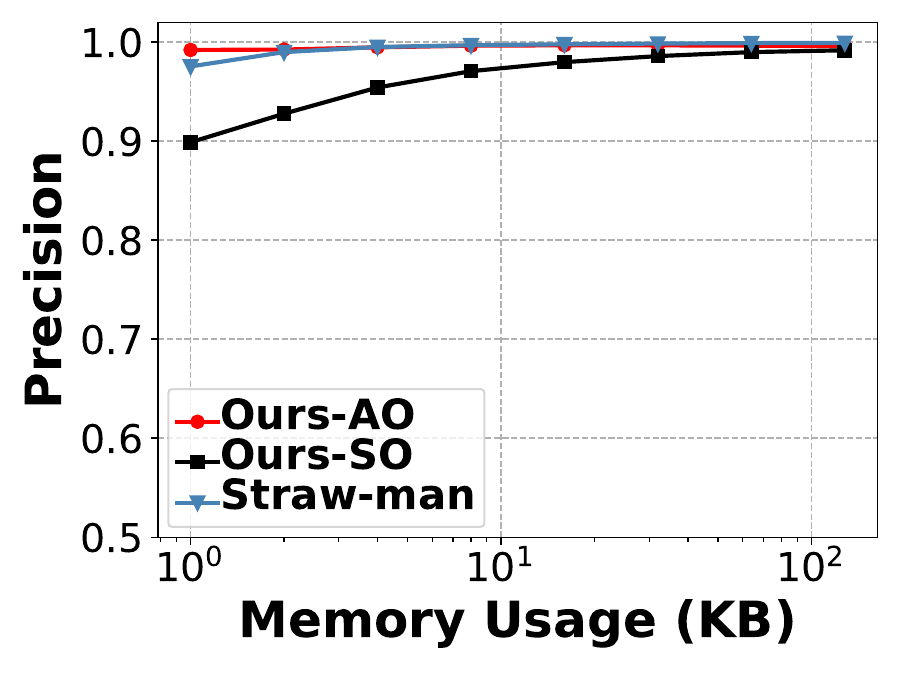}
		\end{center}
		}
		\label{real-PR-CAIDA}
		\end{minipage}
	}
	\subfigure[MAWI]{
		\begin{minipage}[t]{0.225\textwidth}{
		\begin{center}
		\includegraphics[width=\textwidth]{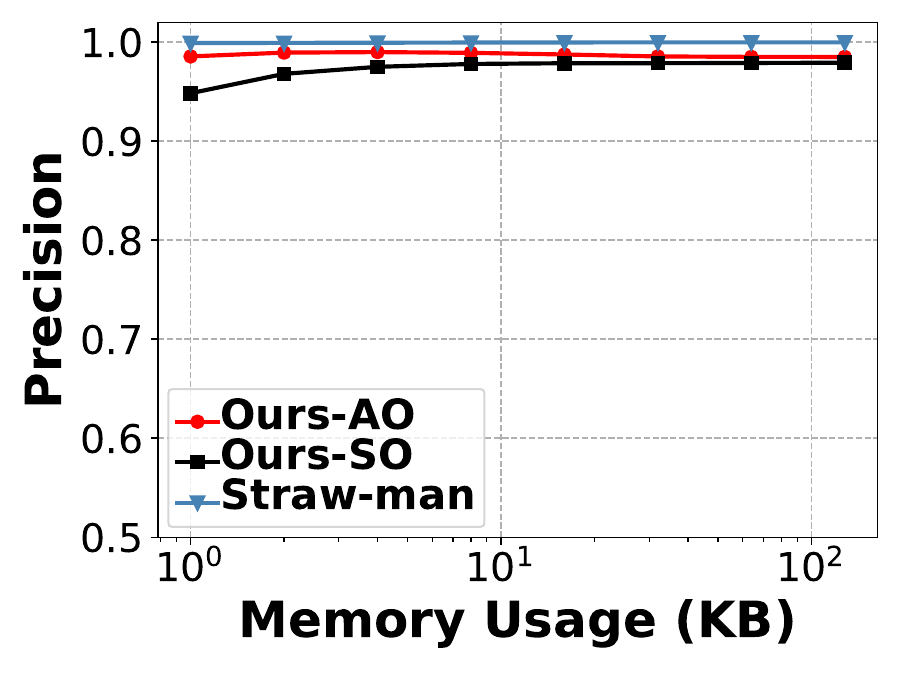}
		\end{center}
		}
		\label{real-PR-MAWI}
		\end{minipage}
	}
	\subfigure[MACCDC]{
		\begin{minipage}[t]{0.225\textwidth}{
		\begin{center}		
		\includegraphics[width=\textwidth]{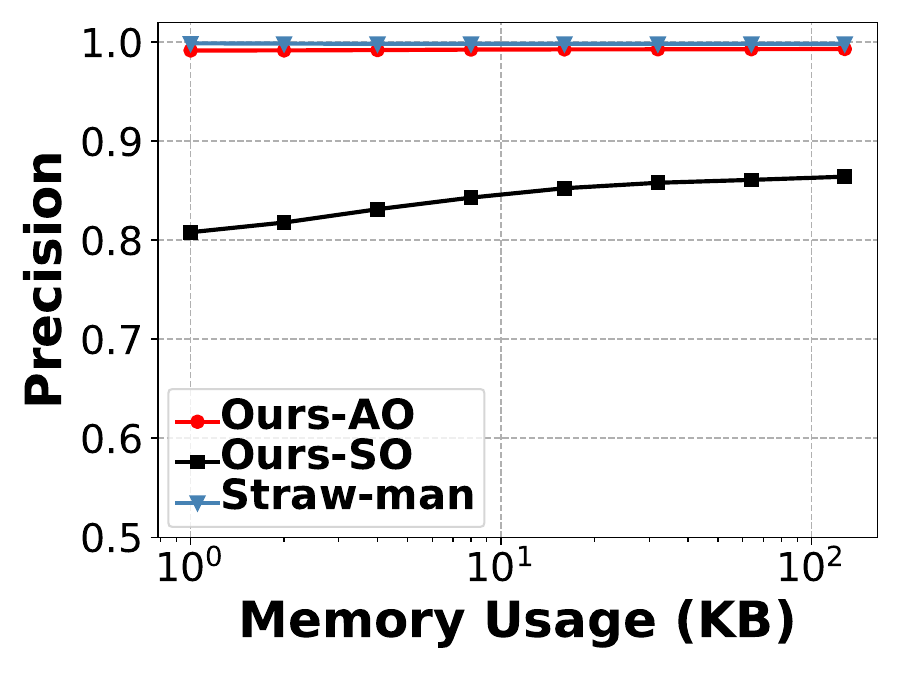}
		\end{center}
		}
		\label{real-PR-MACCDC}
		\end{minipage}
	}
	\subfigure[IMC]{
		\begin{minipage}[t]{0.225\textwidth}{
		\begin{center}
		\includegraphics[width=\textwidth]{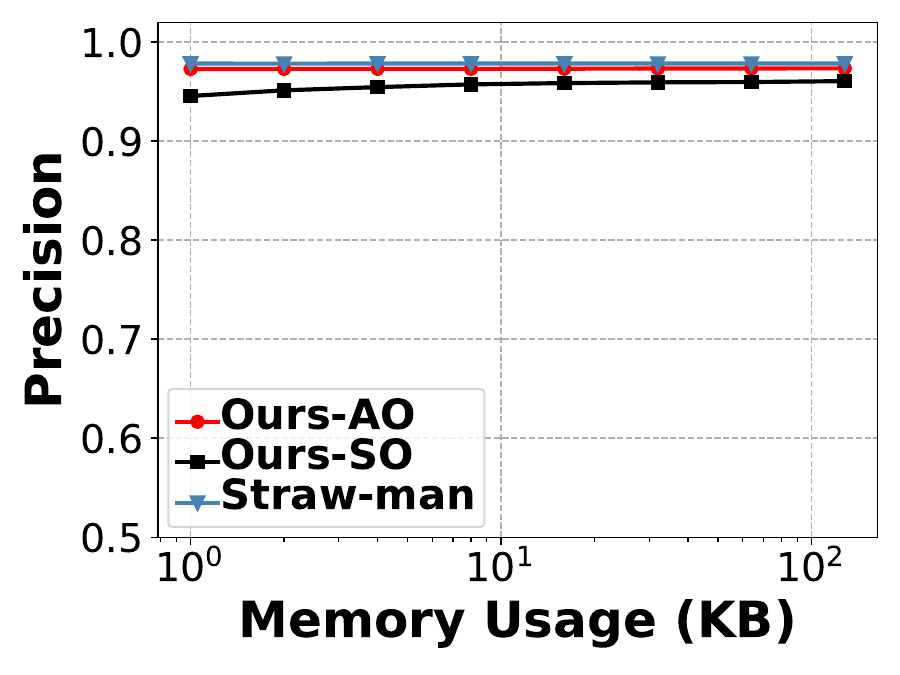}
		\end{center}
		}
		\label{real-PR-IMC}
		\end{minipage}
	}
 %\vspace{-0.1in}
    \caption{PR of algorithms in Real-time Report Scenario on different datasets.}
    \label{real-PR}
    %\vspace{-0.08in}
\end{figure*}

\vspace{0.1in}
\bbb{Effect of \textit{w} on \algo{}-SO (Figure~\ref{para-SOwidth}).} We perform experiments with $w$ ranging from 2 to 16 in \algo{}-SO and observe their performance on various datasets. As shown in the graph, the accuracy of \algo{}-SO first increases and then decreases as $w$ enlarges. 
% The \algo{}-SO with $w=4$ outperforms others and demonstrates high robustness in most cases. 
The \algo{}-SO with $w=8$ performs best.
% on dataset CAIDA, MAWI.
This is attributable to the two counteracting effects on accuracy as $w$ gets larger: the risk of different $SEQ$ colliding in a bucket increases, while the probability that too many large flows crush in a bucket decreases.

% and result in an overflow decreases.
% of the situation that one bucket is overflowed while another bucket still has plenty of cells, which is caused by uneven distribution of items to buckets decided by hash functions. 

% , so we choose $w=4$ as our default parameter.

% A longer fingerprint decrease the risk of hash collision, but takes up more memory. And a shorter fingerprint takes up less memory, but increase the risk of hash collision. We want to see what the optimal value for the number of bits of the fingerprint is. And we do tests on both real-time report and periodic report scenarios on CAIDA and VA stream datasets.

\bbb{Effect of \textit{s/c} on \algo{}-AO (Figure~\ref{para-AOratio}).} We perform experiments with the $suspect/civilian$ memory ratio ranging from $1:7$ to $7:1$ in \algo{}-AO and observe their performance on various datasets. As shown in the graph, the accuracy of \algo{}-AO first rises and then falls when $suspect/civilian$ memory ratio becomes larger. The \algo{}-AO with ratio=$3:5$ performs best.
% on dataset CAIDA, MAWI. 
As we analyze in ~\ref{AO:analysis}, the $suspect$ part is designed to protect the small abnormal flows from being ousted by large normal flows. However, when the $suspect$ part consumes too much memory, we lack sufficient $civilian$ cells to monitor overall flows, leading to a loss of much necessary information.

\bbb{Effect of fingerprint length on \algo{}-AO (Figure~\ref{para-AOfp}).} We conduct experiments with fingerprint lengths ($l_f$) ranging from 2 bits to 16 bits in \algo{}-AO and observe their performance on different datasets. As shown in the graph, the accuracy of \algo{}-AO first increases and then decreases when $l_f$ enlarges. The \algo{}-AO with $l_f=8$ bits performs the best.
% on dataset MACCDC, IMC. 
This is because on one hand, increasing $l_f$ improves the accuracy of matching(see Section ~\ref{OP:FP} for more details) by reducing mistakes caused by $seq$ collisions. On the other hand, a fingerprint with a larger $l_f$ will occupy more memory, decreasing the total number of cells.  

\bbb{Parameter Selection.} In experiments from here on, we set $w=8$,
$l_f=8$ 
and $suspect/civilian$ memory ratio to be $3:5$ because they are the most robust setting considering performances at different datasets. 

% \reviewA{
% \bbb{Parameter adjusting.} As shown in Figure~\ref{fig:adjustingAlgo}, our parameter adjusting algorithm can efficiently adjust parameter to adapt to the dataset and improve accuracy. We cut the data stream into multiple windows and adjust the parameters according to the algorithm described in Section~\ref{sec:parameterAdjusting} at the end of each window. We can see that the $F_1$ Score of all versions of algorithms increase gradually and become stable with more and more windows passing by. 
% }
% \bbb{Analysis.}

% \bbb{Parameter Selection.} In experiments from here on, we set $w=4$ for best performance on all datasets in all scenarios. We select $l_f$ to be 8 bits and $suspect/civilian$ memory ratio to be $2:6$ because they are the most robust setting considering performances at different datasets in different scenarios.  

% according to the elaborate experimental results in Figure~\ref{}.
% \vspace{-0.2in}
% \begin{figure}[!ht]
% 	\centering
%         \subfigure[CAIDA]{
% 		\begin{minipage}[t]{0.225\textwidth}{
% 		\begin{center}
% 		\includegraphics[width=\textwidth]{Figures/exp/fig/CAIDA-tho.pdf}
% 		\end{center}
% 		}
% 		\label{throu-CAIDA}
% 		\end{minipage}
% 	}
%         %
% 	\subfigure[MAWI]{
% 		\begin{minipage}[t]{0.225\textwidth}{
% 		\begin{center}
% 		\includegraphics[width=\textwidth]{Figures/exp/fig/MAWI-tho.pdf}
% 		\end{center}
% 		}
% 		\label{throu-MAWI}
% 		\end{minipage}
% 	}
%  %\vspace{-0.1in}
%     \caption{Throughput of algorithms on different datasets.}
%     \label{throu}
%     % \vspace{-0.20in}
% \end{figure}

\begin{figure*}[!ht]
	\centering
        \subfigure[CAIDA]{
		\begin{minipage}[t]{0.235\textwidth}{
		\begin{center}
		\includegraphics[width=\textwidth]{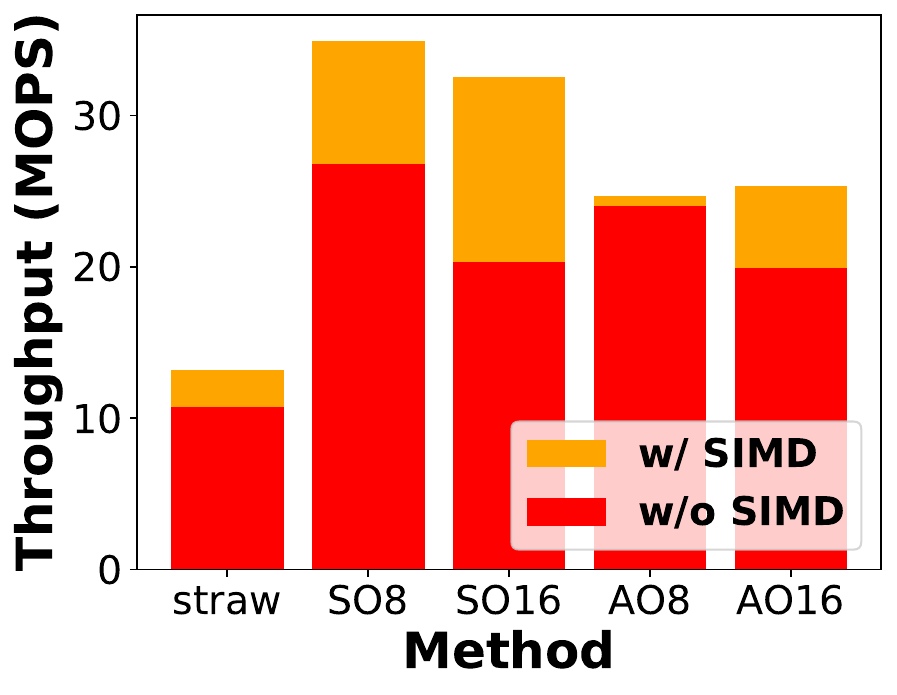}
		\end{center}
		}
		\label{throu-CAIDA}
		\end{minipage}
	}
	\subfigure[MAWI]{
		\begin{minipage}[t]{0.235\textwidth}{
		\begin{center}
		\includegraphics[width=\textwidth]{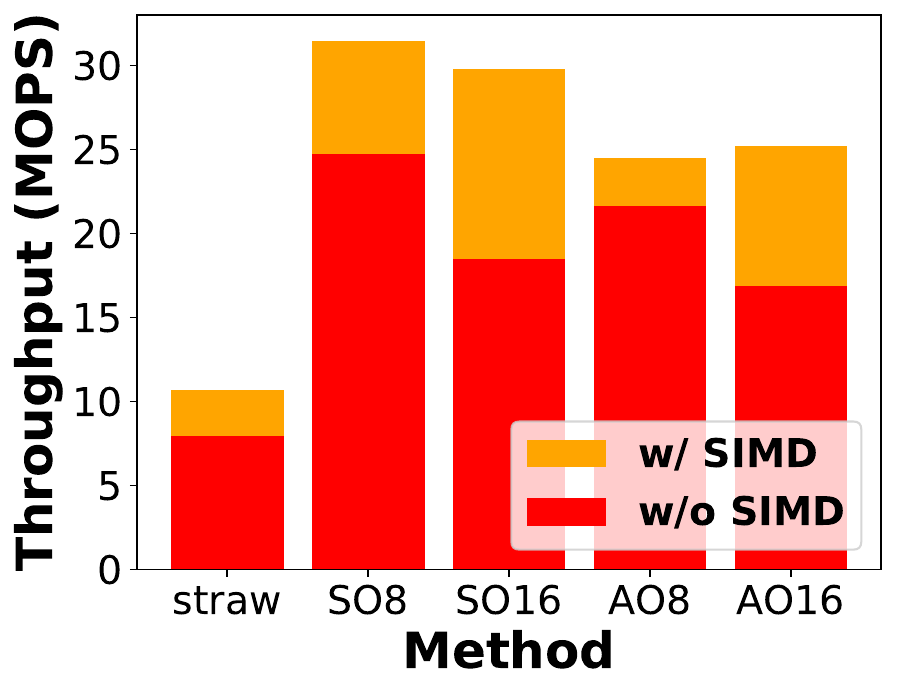}
		\end{center}
		}
		\label{throu-MAWI}
		\end{minipage}
	}
	\subfigure[MACCDC]{
		\begin{minipage}[t]{0.235\textwidth}{
		\begin{center}		
		\includegraphics[width=\textwidth]{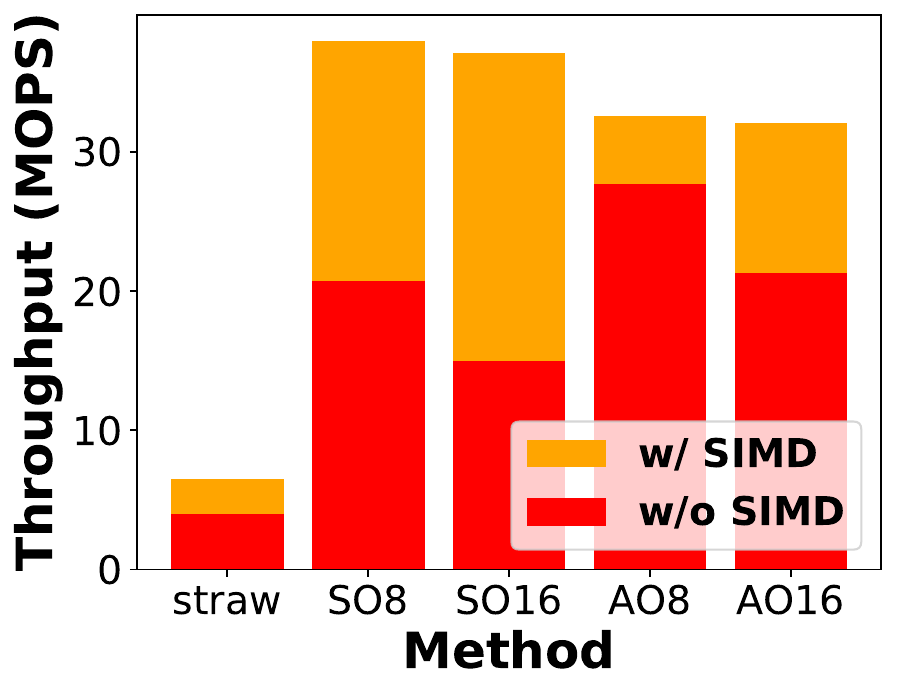}
		\end{center}
		}
		\label{throu-MACCDC}
		\end{minipage}
	}
	\subfigure[IMC]{
		\begin{minipage}[t]{0.235\textwidth}{
		\begin{center}		
		\includegraphics[width=\textwidth]{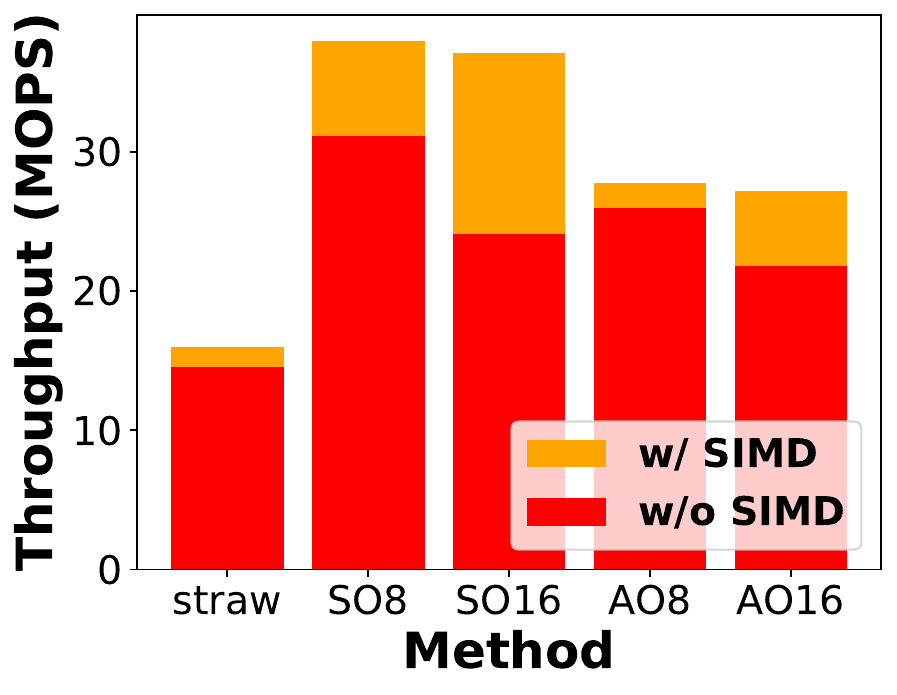}
		\end{center}
		}
		\label{throu-IMC}
		\end{minipage}
	}
 %\vspace{-0.1in}
    \caption{Throughput of algorithms on different datasets.}
    \label{throu}
    %\vspace{-0.08in}
\end{figure*}

\vspace{-0.1in}
\subsection{Experiments on Accuracy}
\label{exp:acc}

\reviewA{
In this section, we compare the accuracy of the Straw-man solution, \algo{}-SO and \algo{}-AO. We analyze the changes in $F_1$ as we alter the memory allocation from 1KB to 128KB.
}

\reviewA{
\bbb{$F_1$-Score (Figure~\ref{real-F1}).} The experiments shows that the $F_1$ of \algo{}-AO and \algo{}-SO are constantly higher than that of the Straw-man solution. We conduct experiments on different datasets and find that on CAIDA, MAWI, MACCDC, and IMC, the $F_1$-Scores of \algo{}-AO are respectively 1.61, 1.50, 1.12, and 1.04 times higher on average than those of the Straw-man solution, while the $F_1$-Scores of \algo{}-SO are on respectively 1.40, 1.38, 1.04, and 1.02 times higher on average than those of the Straw-man solution.
}

\bbb{RR (Figure~\ref{real-RR}).} According to the experiment results, the RR of \algo{}-AO and \algo{}-SO are significantly higher than that of the Straw-man solution. We conduct experiments on different datasets and find that on CAIDA, MAWI, MACCDC, and IMC, the RR of \algo{}-AO are respectively 3.26, 3.85, 1.63, and 1.22 times higher on average than those of Straw-man solution, while the RR of \algo{}-SO are respectively 2.76, 3.20, 1.38, and 1.19 times higher on average than those of the Straw-man solution.

\bbb{PR (Figure~\ref{real-PR}).} We can see from the experiment results that \algo{}-AO and Straw-man possess a similar precision rate with any memory limitation. The precision rate of \algo{}-SO increases rapidly with memory increasing. 
The precision rate of the Straw-man solution is always 1 since it records the flow ID.

\bbb{Analysis.} It is shown in the experiment results above that \algo{}-AO and \algo{}-SO can achieve a much better accuracy than the Straw-man solution. Notably, the Straw-man solution requires about 32 times more memory than \algo{}-AO to achieve the same accuracy. Several key factors contribute to the excellence of our algorithms. First, we utilize the sequence number as the index number to achieve matching, saving the memory of the 13-byte flow ID. This approach poses the challenge of avoiding $seq$ collisions in the huge number of flows within the data stream. We effectively address this issue by grouping and using the fingerprint as assistance. Second, our $civilian$-$suspect$ mechanism efficiently combines a rough monitoring on the overall flows and a meticulous monitoring on the suspicious flows organically, achieving an extraordinary performance. Third, we employ the LRU and LRD replacing policies to keep the most critical information. Furthermore, LRU and LRD are implemented without extra memory or time overhead.

\vspace{-0.1in}
\subsection{Experiments on Processing Speed}
\label{exp:speed}

\reviewA{
We compare the throughput of \algo{}-SO with $w=8/16$, \algo{}-AO with $w=8/16$, and Straw-man. 
}

\reviewA{
\bbb{Throughput (Figure~\ref{throu}).} As the results manifest, in every dataset, the fastest algorithm is \algo{}-SO. Both \algo{}-SO and \algo{}-AO are significantly faster than the Straw-man solution, with or without SIMD. SIMD brings more improvement when the number of cells in a bucket is larger.
On dataset CAIDA, MAWI,  MACCDC and IMC, the 
throughput of \algo{}-SO is respectively 2.57, 3.03, 5.19 and 2.18 times higher on average than those of the Straw-man solution, while the throughput of \algo{}-AO is respectively 2.10, 2.51, 5.60 and 1.69 times higher than the Straw-man solution.
the fastest algorithm is \algo{}-SO, being 1.6 times faster than the Straw-man solution. The speed of \algo{}-AO and the Straw-man solution are similar. 
% The corresponding CPU utilization in the experiments is above $99\%$ the whole time.
% The corresponding CPU utilization of the experiments are shown in Figure~\ref{fig:cpu}.
}
% , followed by straw-man, and the slowest is \algo{}-AO.(Except for MACCDC, where \algo{}-AO is faster than straw-man ) And the memory size of data structure does not affect the throughput much.

% \bbb{Throughput (Figure~\ref{throu}).} As the results manifest, in every dataset, the fastest algorithm is \algo{}-SO. Both \algo{}-SO and \algo{}-AO are significantly faster than the Straw-man solution, with or without SIMD. SIMD brings more improvement when the number of cells in a bucket is larger.
% On dataset CAIDA and MAWI, the 
% throughput of \algo{}-SO is respectively 2.57 and 3.03 times higher on average than those of the Straw-man solution, while the throughput of \algo{}-AO is respectively 2.10 and 2.51 times higher than the Straw-man solution.

% \algo{}-SO and \algo{}-AO always have the same order of magnitude of time cost as the straw-man solution when processing different amounts of items.

\bbb{Analysis.} 
% The reason why \algo{}-SO is faster than Straw-man is that \algo{}-SO does not need to access flow ID and it utilizes SIMD operations for further acceleration, allowing for querying and modification of all the cells in a bucket with one operation. While straw-man needs to access every cell in the bucket one by one. 
\algo{}-SO is significantly faster than Straw-man due to several key advantages: {\textbf{(1)} Our operations are neat and has good spatial locality. We only need to calculate at most three hash numbers for every incoming item. Once a bucket is chosen, all the following operations are conducted in the bucket. We also keep the size of a group small, so the memory taken by a bucket is small. In contrast, there are at least three hashing calculations for every incoming item in the Straw-man solution. Besides, the eviction operation in the Straw-man solution has a bad spatial locality.} \reviewB{ \textbf{(2)} Straw-man needs to access the fingerprint for matching and access the sequence number for detecting flow gaps. While \algo{}-SO only needs to access the sequence number to accomplish both matching and detecting flow gaps.} 
% Second, it capitalizes on SIMD operations to expedite acceleration. This enables \algo{}-SO to query and modify all cells within a bucket simultaneously as a single operation. In contrast, Straw-man operates at a slower pace, as it requires individual access to each cell in the bucket and it needs to first access flow ID for matching before accessing the sequence number for \textit{flow gap} detection.

\begin{figure}[!ht]
	\centering
        \subfigure[fingerprint]{
		\begin{minipage}[t]{0.225\textwidth}{
		\begin{center}
		\includegraphics[width=\textwidth]{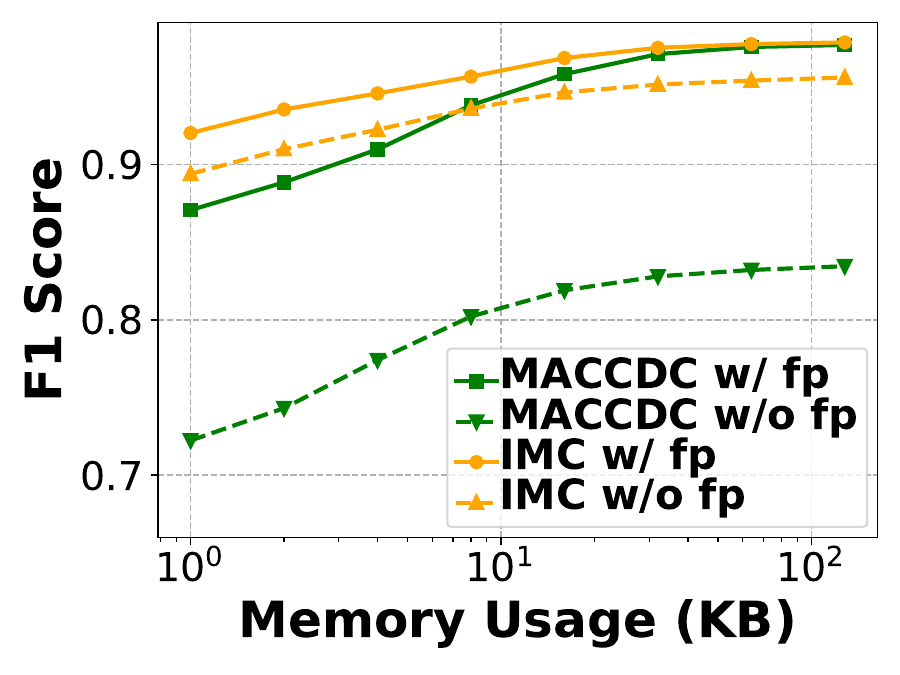}
		\end{center}
		}
		\label{fig:optim:fingerprint}
		\end{minipage}
	}
        \subfigure[sequence number randomizing]{
		\begin{minipage}[t]{0.225\textwidth}{
		\begin{center}
		\includegraphics[width=\textwidth]{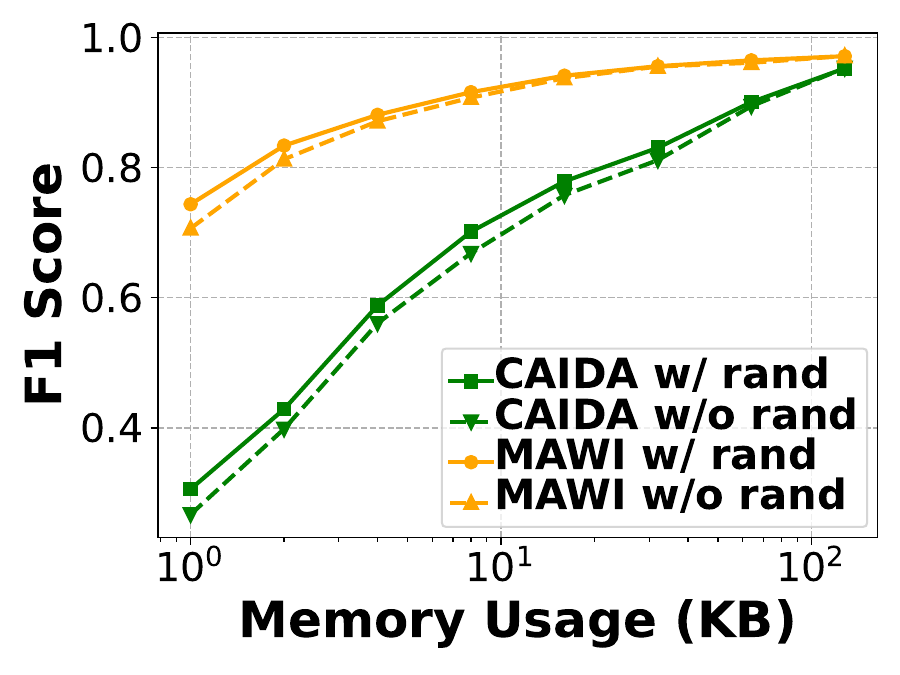}
		\end{center}
		}
		\label{fig:optim:rand}
		\end{minipage}
	}

    \caption{Effect of optimizations on different datasets.}
    \label{fig:optim}
    % \vspace{-0.20in}
\end{figure}

\subsection{Experiments on optimizations}
\label{exp:opt}
In this section, we show the improvement of accuracy provided by sequence number randomizing and fingerprint.

{
\bbb{Effect of fingerprint (Figure~\ref{fig:optim:fingerprint}).} The experiment results show that fingerprint significantly enhances the accuracy on dataset MACCDC and IMC. Specifically, \algo{}-AO with fingerprint performs 1.18 times better than that without fingerprint.
}

{
\bbb{Effect of sequence number randomizing (Figure~\ref{fig:optim:rand}).} The experiment results show that sequence number randomizing improves the accuracy on dataset CAIDA and MAWI. Specifically, \algo{}-AO with sequence number randomizing performs 1.04 times better than that without sequence number randomizing.
}

\subsection{Experiments on Pattern of Flow Gaps}
In this section, we alter the ratio $r$ of abnormal flows in a time window and the parameter $b$ determining the probability of \textit{major gap} taking place in an abnormal flow in the synthetic item loss, as described in Section~\ref{exp:dataset}, in order to test the robustness of \algo{}. 
% We compare the accuracy of \algo{}-SO, \algo{}-AO and Straw-man solution in both Real-time Report and Periodic Report Scenarios.

\begin{figure}[!ht]
	\centering
        \subfigure[The effect of $r$]{
		\begin{minipage}[t]{0.225\textwidth}{
		\begin{center}
		\includegraphics[width=\textwidth]{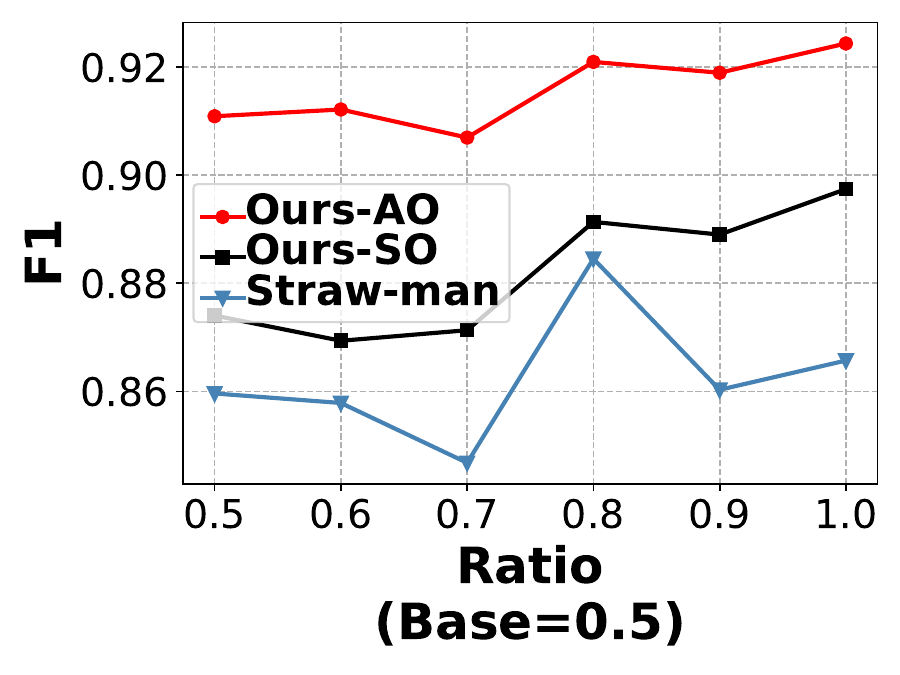}
		\end{center}
		}
		\label{pat-ratio}
		\end{minipage}
	}
	\subfigure[The effect of $b$]{
		\begin{minipage}[t]{0.225\textwidth}{
		\begin{center}
		\includegraphics[width=\textwidth]{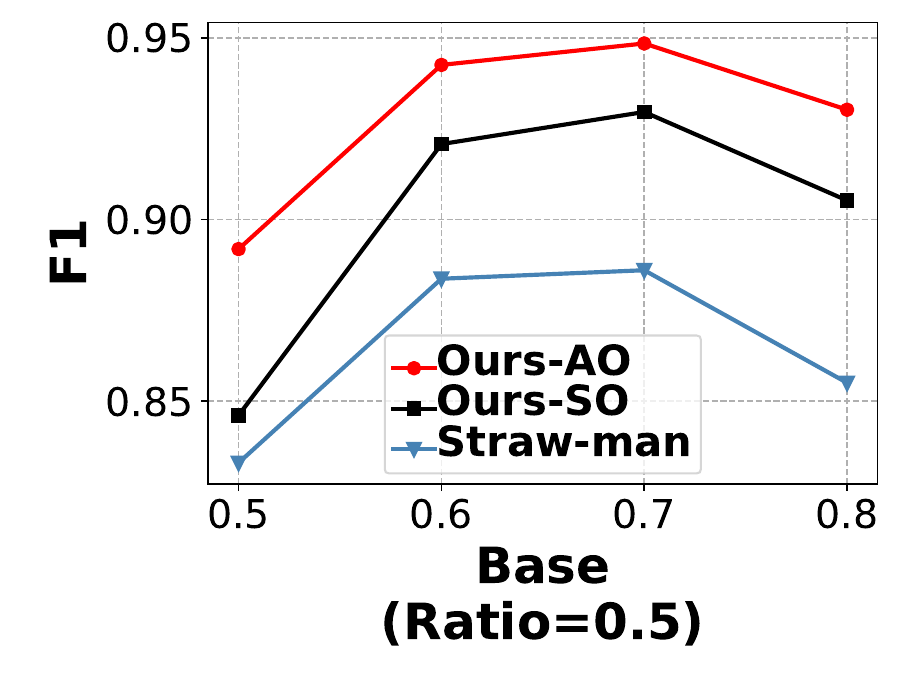}
		\end{center}
		}
		\label{pat-base}
		\end{minipage}
	}
	%
	% \subfigure[IMC]{
	% 	\begin{minipage}[t]{0.225\textwidth}{
	% 	\begin{center}		
	% 	\includegraphics[width=\textwidth]{Figures/exp/tmp.pdf}
	% 	\end{center}
	% 	}
	% 	\label{pat-ratio-IMC-RMSE}
	% 	\end{minipage}
	% }
 %\vspace{-0.1in}
    \caption{The effect of $r$ and $b$
    % on the performance of algorithms 
    on IMC dataset.}
    \label{pat}
    % \vspace{-0.10in}
\end{figure}

\bbb{Effect of $r$ on performance (Figure~\ref{pat-ratio}).} As shown in the figures, \algo{}-AO and \algo{}-SO perform much better than Straw-man. The $F_1$-Scores of \algo{}-AO and \algo{}-SO are on average 1.07 and 1.03 times higher than that of Straw-man. 

% The MAE of Straw-man is on average 2.04 and 1.34 times higher than that of \algo{}-AO and \algo{}-SO.

\bbb{Effect of $b$ on performance (Figure~\ref{pat-base}).} The experiment results show that \algo{}-AO and \algo{}-SO outperform the Straw-man solution. The $F_1$-Scores of \algo{}-AO and \algo{}-SO are on average 1.08 and 1.05 times higher than that of Straw-man. 

% The MAE of Straw-man is on average 2.85 and 1.70 times higher than that of \algo{}-AO and \algo{}-SO.

\bbb{Analysis.} The results above demonstrate that \algo{} can deal with various patterns of \textit{flow gap}. It has high accuracy and robustness, even in extreme network circumstances.

\end{sloppypar}
    \vspace{-0.1in}

\presec
\section{Conclusion}
\label{sec:conclusion}
\postsec

In this paper, we formalize the task of monitoring the \textit{value} variation of items in the data stream, which is a vacancy in research fields. A scenario corresponding to this task is the real-time detection of the \textit{flow gaps} in the data stream. To solve this problem, we propose \algo, which has two key ideas: \textit{similarity absorption} technique and the $civilian$-$suspect$ mechanism. We conduct extensive experiments on four real-world datasets and design a method of synthetic item loss. Our experiment results show that compared with the Straw-man solution, our algorithm can achieve the same accuracy with 1/32 memory overhead.

% improves the accuracy up to six times.

% In this paper, we propose the Double-Anonymous sketch, which is the first work that achieves \fairness{} of global top-$K$. 
% We theoretically prove that the Double-Anonymous sketch achieves both unbiasedness and double-anonymity, so as to achieve \fairness{}. 
% %We also derive an error bound strictly tighter than the well-known Count-Min sketch bound.
% We conduct extensive experiments on three real and one synthetic dataset. 
% Our experimental results show that compared with the state-of-the-art, our algorithm improves the accuracy 129 times.

    \section{Acknowledgments}
\begin{acks}
% We would like to thank the anonymous reviewers for their valuable suggestions. 
%
% This work is supported by Key-Area Research and Development Program of Guangdong Province 2020B0101390001, and National Natural Science Foundation of China (NSFC) (No. U20A20179).
This work is supported by supported by Beijing Natural Science Foundation (Grant No. QY23043). The authors gratefully acknowledge the foundation's financial assistance.
\end{acks} 

    % \clearpage
    \bibliographystyle{unsrt}
    \bibliography{InputFiles/reference}	

    % \clearpage
    % \appendix
    % \input{tex/7.appendix}

\end{document}